\documentclass[aps,prd,preprint,nofootinbib,floatfix,10pt]{revtex4-2}
\usepackage[T1]{fontenc}
\usepackage{lmodern}
\usepackage[utf8]{inputenc}
\usepackage{amsmath,amssymb,amsfonts,mathtools}
\usepackage{microtype}
\usepackage{xcolor}
\usepackage{enumitem}
\usepackage{amsthm}
\usepackage{tikz}
\usetikzlibrary{arrows.meta}
\usepackage{hyperref}

\usepackage{mathrsfs}
\hypersetup{
 colorlinks=true,
 linkcolor=blue,
 citecolor=blue,
 urlcolor=blue,
 pdftitle={Nonlinear Degrees of Freedom of Born-Infeld New Massive Gravity},
 pdfauthor={Zeynep Su Koc and Bayram Tekin},
 pdfsubject={Nonlinear Hamiltonian and Dirac constraint analysis of Born-Infeld New Massive Gravity}
}
\usepackage{etoolbox}

\makeatletter
\ifdefined\frontmatter@abstract@produce
\patchcmd{\frontmatter@abstract@produce}
  {\penalty-200\relax}
  {\penalty10000\relax}
  {}
  {\PackageWarning{BINMG}{Could not patch abstract page break}}
\fi
\makeatother
\newcommand{\dd}{\mathrm d}

\newcommand{\weak}{\approx}

\newtheorem{theorem}{Theorem}

\newtheorem{proposition}{Proposition}

\theoremstyle{remark}

\begin{document}

\title{Nonlinear Degrees of Freedom of Born--Infeld New Massive Gravity}

\author{Zeynep Su Koc}
\email{su.koc@bilkent.edu.tr}
\author{Bayram Tekin}
\email{bayram.tekin@bilkent.edu.tr}
\affiliation{Department of Physics, Bilkent University, 06800 Ankara, T\"urkiye}

\date{\today}

\begin{abstract}
Born--Infeld New Massive Gravity (BINMG) is an all-order curvature extension
of New Massive Gravity in three dimensions, whose linearized spectrum around
its unique maximally symmetric vacuum contains the two polarizations of a massive
spin-two field.  Whether an additional scalar degree of freedom appears in
the full nonlinear theory has remained an open question.  We perform a
nonlinear canonical analysis of BINMG using its auxiliary-metric
formulation.  On the regular branch connected to the maximally symmetric
vacuum, the Dirac algorithm yields three first-class diffeomorphism
constraints and a second-class pair of scalar constraints.  The resulting
count gives exactly two local degrees of freedom per spatial point. Thus, nonlinear BINMG propagates no additional Boulware--Deser-type scalar mode on
this branch.
\end{abstract}
\maketitle

\section{Introduction}
\label{sec:introduction}

To gain some understanding of gravity, and of quantum gravity in
particular, one may offer a gambit, as in chess: give up a dimension in
exchange for the initiative.  The sacrifice is real.  In three spacetime
dimensions, the Weyl tensor vanishes identically, the Riemann tensor is
determined algebraically by the Ricci tensor, and pure Einstein gravity therefore has no local propagating graviton; in Hamiltonian language, the
six phase-space variables of the spatial metric per point are removed
entirely by one normal and two tangential first-class diffeomorphism
constraints.  But so is the compensation: the theory still has black holes
\cite{BTZ}, and it becomes soluble in regimes where its four-dimensional
counterpart is intractable \cite{Witten}.  Three dimensions also admit
several consistent ways of giving gravity local massive excitations.  One
consequence of the gambit should be kept in mind: since the massless
graviton has been conceded, any local bulk degree of freedom of a
three-dimensional higher-curvature theory must come from the deformation
itself.

The first nonlinear massive gravity theory in three dimensions
was Topologically Massive Gravity (TMG), which propagates one parity-odd
massive helicity \cite{DeserJackiwTempleton}.  It is a dynamical theory of
gravity, but with a single local degree of freedom, it does not directly
resemble four-dimensional parity-even Einstein gravity, which has two
massless graviton polarizations.  Almost thirty years later, New Massive
Gravity (NMG) was discovered: a parity-even theory propagating the two
polarizations of a massive spin-two field \cite{BHT2009,BHTMore}. The special feature of
NMG is that its curvature-squared term is tuned so that the scalar mode, which would be present in a generic curvature-squared theory, is removed.  The nonlinear question, however, is not answered by the linearized spectrum alone \cite{GulluTek,GSTCanonical,Tekin2016}.  A theory may have the correct perturbative spin-two spectrum around a special background and still possess an additional scalar mode in the full nonlinear phase space.  This is the Hamiltonian form of the familiar Boulware--Deser issue in nonlinear massive gravity \cite{BoulwareDeser}. For related difficulties that arise when a Fierz--Pauli mass term is added
to TMG, see \cite{MandM}.

Born--Infeld New Massive Gravity (BINMG) is an all-order curvature completion
of NMG \cite{GSTBINMG}. Its small-curvature expansion reproduces the NMG
combination at quadratic order, while its cubic and higher-order terms are
compatible with AdS/CFT requirements \cite{Sinha,Gullu:2010st,Paulos0}.
In contrast to NMG, whose maximally symmetric vacuum equation is quadratic
and generically has two branches, BINMG has a unique maximally symmetric
vacuum.  Around this vacuum, the perturbative
spectrum is the expected massive spin-two spectrum (see the Ph.D. thesis
\cite{TahsinTez} and the references therein). We do not repeat that
linearized analysis here.  Our question is the nonlinear one: does the full
Dirac constraint structure \cite{Dirac,Castellani,HenneauxTeitelboim} of
BINMG leave exactly two local degrees of freedom?

For NMG itself, the nonlinear question has been answered affirmatively:
complete constraint analyses in metric and first-order variables
\cite{BlagojevicCvetkovic,SadeghShirzad}, the Hamiltonian form of
three-dimensional massive gravity \cite{HohmRouthTownsendZhang}, and the
Chern--Simons--like framework \cite{CSlike,AfsharBergshoeffMerbis} all
conclude that NMG carries two local degrees of freedom nonlinearly.  Those
methods exploit the fact that NMG admits a first-order reformulation
polynomial in a finite collection of one-form fields; to our knowledge, no
Chern--Simons--like reformulation of BINMG exists because the Born--Infeld
determinant is not polynomial in such variables, so the all-order theory
falls outside that framework. For BINMG, however, the evidence for nonlinear ghost freedom has so far
been indirect: Paulos and Tolley \cite{PaulosTolley} obtained the
theory as a scaling limit of a ghost-free bigravity model, which makes
nonlinear ghost freedom plausible but leaves the intrinsic constraint
structure of BINMG itself unexamined---and limits of ghost-free theories need
not preserve the constraint count automatically (see
Ref.~\cite{deRhamNMG} for how delicate such statements are in
three-dimensional massive gravity).  A direct Dirac analysis of BINMG in
its own variables, with the branch assumptions spelled out, has been
missing.  This paper supplies it.

We answer this question using the auxiliary-metric formulation of BINMG, which was recently provided in \cite{TekinAuxMetric2026}.  This
formulation is algebraically equivalent to the determinant action on the
regular branch where the auxiliary metric is nondegenerate.  The useful
Hamiltonian variable is the densitized inverse auxiliary metric
\begin{equation}
        P^{\mu\nu}
        :=
        \sqrt{-q}\,q^{\mu\nu}.
        \label{eq:Pmunu-intro}
\end{equation}
We decompose \(P^{\mu\nu}\) with respect to the ADM foliation of the physical
metric.  Its spatial block will be denoted by \(S^{ij}\), its mixed
normal-spatial component by \(v^i\), and its normal-normal component by
\(\nu\).  The key step is to write the full ADM form of the auxiliary
action in these variables before the Legendre transform.  The constraint
structure is then visible: \(v^i\) is algebraic on the branch where
\(S^{ij}\) is invertible, while \(\nu\) appears linearly and ends up as the
multiplier of a scalar constraint.

After eliminating the algebraic mixed auxiliary component, the dynamical
canonical variables are
\begin{equation}
        (\gamma_{ij},\rho^{ij}),
        \qquad
        (S^{ij},p_{ij}) .
        \label{eq:canonical-block-intro}
\end{equation}
Since both \(\gamma_{ij}\) and \(S^{ij}\) are symmetric two-dimensional
tensors, this gives twelve canonical phase-space variables per spatial point.
General covariance supplies three first-class diffeomorphism generators on
this reduced dynamical phase space.  They remove six phase-space dimensions,
but this is not enough to reduce the theory to two local degrees of freedom.
Without an additional scalar restriction, the nonlinear theory would have
three.  

The remaining question is therefore the auxiliary scalar sector.  The normal
auxiliary component \(\nu\) imposes a scalar constraint, which we denote by
\(\mathcal C_0\).  Preserving this constraint in time gives a secondary scalar
constraint, denoted by \(\mathcal C_1\).  We show that, on the regular branch
connected to the maximally symmetric vacuum, the pair
\((\mathcal C_0,\mathcal C_1)\) is second class.  Equivalently, preserving
\(\mathcal C_1\) fixes the auxiliary normal variable \(\nu\), not
the physical lapse.  The physical lapse and shift remain arbitrary, as required by
diffeomorphism invariance.

The local degree-of-freedom count on this branch is then
\begin{equation}
        N_{\rm dof}
        =
        \frac12\left(12-2\times3-2\right)
        =
        2.
        \label{eq:intro-dof-count}
\end{equation}
Thus BINMG propagates exactly two local degrees of freedom on the regular
vacuum-connected branch, which is the nonlinear continuation of the massive
spin-two sector.  The result is a branch theorem: degenerate auxiliary
metrics, noninvertible spatial auxiliary blocks, singular ADM charts, and
configurations where the scalar Dirac matrix loses rank are left for a
separate analysis.

The paper is organized as follows.  Section~\ref{sec:BINMG} recalls the
auxiliary formulation.  Section~\ref{sec:nonlinear-hamiltonian-analysis}
contains the Hamiltonian analysis: the ADM decomposition, the Legendre
transform, the diffeomorphism sector, the elimination of \(v^i\), and the
scalar Dirac chain generated by \(\nu\).  Section~\ref{sec:degree-of-freedom-theorem}
states the theorem, and Sec.~\ref{sec:scope-and-discussion} discusses its
domain of validity.

\section{Conventions and notation}
\label{sec:conventions}

We work in three spacetime dimensions with metric signature \((-++)\).
Greek indices \(\mu,\nu,\ldots\) denote spacetime components, while Latin
indices \(i,j,\ldots\) denote components tangent to a spatial slice.  Our
curvature convention is
\begin{equation}
 R^\mu{}_{\nu\rho\sigma}
 =
 \partial_\rho\Gamma^\mu_{\nu\sigma}
 -
 \partial_\sigma\Gamma^\mu_{\nu\rho}
 +
 \Gamma^\mu_{\rho\lambda}\Gamma^\lambda_{\nu\sigma}
 -
 \Gamma^\mu_{\sigma\lambda}\Gamma^\lambda_{\nu\rho},
\end{equation}
with
\begin{equation}
        R_{\nu\sigma}=R^\mu{}_{\nu\mu\sigma},
        \qquad
        R=g^{\nu\sigma}R_{\nu\sigma},
        \qquad
        G_{\mu\nu}=R_{\mu\nu}-\frac12 g_{\mu\nu}R .
\end{equation}
With this convention a round sphere has positive scalar curvature.
Symmetrization has unit weight,
\begin{equation}
        T_{(\mu\nu)}
        =
        \frac12\left(T_{\mu\nu}+T_{\nu\mu}\right).
\end{equation}

The physical metric is decomposed in ADM form
\cite{ADM,Gourgoulhon} as
\begin{equation}
 \dd s_g^2
 =
 -N^2\dd t^2
 +
 \gamma_{ij}
 \left(\dd x^i+N^i\dd t\right)
 \left(\dd x^j+N^j\dd t\right).
 \label{eq:ADMmetric-explained}
\end{equation}
Here \(N\) is the lapse, \(N^i\) is the shift, and \(\gamma_{ij}\) is the
induced metric on the spatial slice \(\Sigma_t\).  We use
\begin{equation}
        N_i:=\gamma_{ij}N^j .
\end{equation}
The future-directed unit normal is chosen as
\begin{equation}
        n_\mu=(-N,0,0),
        \qquad
        n^\mu=\frac1N(1,-N^i),
        \label{eq:normal-explained}
\end{equation}
and the extrinsic curvature is\footnote{Our convention for the sign of the extrinsic curvature is opposite to the one in \cite{Gourgoulhon}.}
\begin{equation}
 K_{ij}
 =
 \frac1{2N}
 \left(
 \dot\gamma_{ij}
 -
 D_iN_j
 -
 D_jN_i
 \right),
 \qquad
 K=\gamma^{ij}K_{ij}.
 \label{eq:K-convention-early}
\end{equation}
With this sign convention, an expanding slicing has \(K>0\).  Spacetime
covariant derivatives are denoted by \(\nabla_\mu\), while \(D_i\) denotes
the spatial covariant derivative compatible with \(\gamma_{ij}\).  The
intrinsic scalar curvature of \(\Sigma_t\) is written as \({}^{(2)}R\).

We use Dirac's weak equality symbol \(\weak\) \cite{Dirac} for equations that
hold only on the constraint surface.  Equations written with an ordinary
equality sign are strong identities in the relevant phase-space region.  We
use \(\doteq\) for equality up to boundary terms.  Unless boundary terms are
displayed explicitly, all integrations by parts in the bulk Hamiltonian
analysis are understood with compactly supported variations or smearings, so
that no surface term contributes.

The auxiliary formulation uses an auxiliary metric \(q_{\mu\nu}\), but the
canonical analysis is simplest in terms of its densitized inverse
\(P^{\mu\nu}=\sqrt{-q}\,q^{\mu\nu}\), whose ADM components are introduced in
Sec.~\ref{sec:ADM-auxiliary-action}.

\section{New Massive Gravity and the nonlinear question}
\label{sec:3d-massive-gravity}

Higher-curvature deformations of three-dimensional Einstein gravity can
propagate local modes because they enlarge the dynamical phase space.
New Massive Gravity (NMG) is the simplest parity-even example.  Its action is
\begin{equation}
        I_{\rm NMG}
        =
        \frac{1}{\kappa^2}
        \int\dd^3x\sqrt{-g}
        \left[
        \sigma R
        -2\lambda_0m^2
        +
        \frac1{m^2}
        \left(
        R_{\mu\nu}R^{\mu\nu}
        -
        \frac38 R^2
        \right)
        \right].
        \label{eq:NMG-action}
\end{equation}
The curvature-squared combination is special.  A generic curvature-squared
theory in three dimensions would propagate, in addition to the massive
spin-two modes, a scalar mode.  The NMG combination removes this scalar at
the linearized level and leaves the two polarizations of a massive spin-two
field \cite{BHT2009,BHTMore,GSTCanonical,Tekin2016}.

This linearized result is not a nonlinear Hamiltonian theorem.
Diffeomorphism invariance is present in both NMG and BINMG, but it does not
fix the number of local degrees of freedom once the dynamical phase space
has been enlarged.  In pure three-dimensional Einstein
gravity the metric phase space has six variables per spatial point, and the
three first-class diffeomorphism generators remove all of them.  In a
massive higher-curvature theory, an equivalent first-order or auxiliary
formulation contains additional canonical variables
\cite{GSTCanonical,DogruThesis}.  The same three gauge symmetries then remove
the gauge redundancy, but not necessarily all of the additional phase-space
directions.

A massive spin-two field in three dimensions carries two local degrees of
freedom.  After the three first-class diffeomorphism generators have been
taken into account, the nonlinear theory must therefore contain a further
constraint that removes the would-be scalar.  In the auxiliary formulation
used below, the question is whether the normal auxiliary component produces
a second-class scalar pair.  It does, on the regular branch connected to the
maximally symmetric vacuum.

\section{Born--Infeld New Massive Gravity and its auxiliary form}
\label{sec:BINMG}

The Born--Infeld extension of New Massive Gravity is defined by \cite{GSTBINMG}
\begin{equation}
        I_{\rm BINMG}
        =
        -\frac{4m^2}{\kappa^2}
        \int \dd^3x
        \left[
        \sqrt{
        -\det\left(
        g_{\mu\nu}
        +
        \frac{\sigma}{m^2}G_{\mu\nu}(g)
        \right)}
        -
        \beta\sqrt{-g}
        \right],
        \label{eq:BINMG-action-background}
\end{equation}
where
\begin{equation}
        \beta:=1-\frac{\lambda_0}{2},
        \qquad
        m^2>0,
        \qquad
        \sigma=\pm1 .
        \label{eq:beta-def}
\end{equation}
Let
\begin{equation}
        A_{\mu\nu}(g)
        :=
        g_{\mu\nu}
        +
        \frac{\sigma}{m^2}G_{\mu\nu}(g).
        \label{eq:A-background-old}
\end{equation}
For the determinant action we restrict to the branch on which
\(A_{\mu\nu}\) is nondegenerate and has Lorentzian signature continuously
connected to that of \(g_{\mu\nu}\).  On this branch the square root in
\eqref{eq:BINMG-action-background} is real.  This nondegeneracy condition
will be one of the regularity assumptions entering the Hamiltonian analysis
below.

The small-curvature expansion of \eqref{eq:BINMG-action-background} begins as
\begin{equation}
 I_{\rm BINMG}
 =
 \frac1{\kappa^2}
 \int\dd^3x\sqrt{-g}
 \left[
 \sigma R
 -
 2\lambda_0m^2
 +
 \frac1{m^2}
 \left(
 R_{\mu\nu}R^{\mu\nu}
 -
 \frac38R^2
 \right)
 +
 O\!\left(\frac{R^3}{m^4}\right)
 \right].
 \label{eq:BINMG-to-NMG-explicit}
\end{equation}
Thus BINMG reproduces NMG at quadratic order in the curvature.  We quote
the expansion only to identify the theory; it plays no role in what
follows.

The determinant form is not convenient for a direct Dirac analysis.  We
therefore use the auxiliary-metric formulation introduced and studied in \cite{TekinAuxMetric2026}.  One introduces an independent symmetric
tensor \(q_{\mu\nu}\) and writes
\begin{equation}
        I_{\rm aux}[g,q]
        =
        -
        \frac{4m^2}{\kappa^2}
        \int \dd^3x
        \left\{
        \frac12\sqrt{-q}
        \left[
        q^{\mu\nu}A_{\mu\nu}(g)-1
        \right]
        -
        \beta\sqrt{-g}
        \right\}.
        \label{eq:q-aux-action-main}
\end{equation}
Here \(q^{\mu\nu}\) is the inverse of \(q_{\mu\nu}\).  Since no derivatives
of \(q_{\mu\nu}\) occur, its equation of motion is algebraic.  Varying
\eqref{eq:q-aux-action-main} with respect to \(q^{\mu\nu}\) gives
\begin{equation}
        A_{\mu\nu}
        -
        \frac12q_{\mu\nu}
        \left(q^{\alpha\beta}A_{\alpha\beta}-1\right)
        =
        0 .
        \label{eq:q-equation-before-trace}
\end{equation}
Taking the trace with \(q^{\mu\nu}\) gives
\begin{equation}
        q^{\mu\nu}A_{\mu\nu}=3,
\end{equation}
and substituting this back into
\eqref{eq:q-equation-before-trace} gives
\begin{equation}
        q_{\mu\nu}=A_{\mu\nu}(g).
        \label{eq:q-equals-A}
\end{equation}
Therefore \(I_{\rm aux}\) is algebraically equivalent to the determinant
action on the nondegenerate branch.  The full list of regularity conditions
is given in Sec.~\ref{sec:vacuum-scalar-rank}.

For the Hamiltonian analysis it is better not to use \(q_{\mu\nu}\) itself,
but its densitized inverse,
\begin{equation}
        P^{\mu\nu}
        :=
        \sqrt{-q}\,q^{\mu\nu}.
        \label{eq:P-def-pedagogical}
\end{equation}
In three spacetime dimensions,
\begin{equation}
        \sqrt{-q}=-\det P,
\end{equation}
with the sign corresponding to the Lorentzian branch.  Up to the overall
constant factor \(-4m^2/\kappa^2\), the auxiliary Lagrangian density can then
be written as
\begin{equation}
        \mathcal L
        =
        \frac12P^{\mu\nu}g_{\mu\nu}
        +
        cP^{\mu\nu}G_{\mu\nu}(g)
        +
        \frac12\det P
        -
        \beta\sqrt{-g},
        \qquad
        c:=\frac{\sigma}{2m^2}.
        \label{eq:aux-P-action-expanded-pedagogical}
\end{equation}
This is the starting point for the nonlinear Dirac analysis.

One feature of \eqref{eq:aux-P-action-expanded-pedagogical} drives the
whole analysis.  The density is linear in \(P^{\mu\nu}\) except for
\(\det P\), and a determinant is multilinear in the rows of its argument.
After the ADM split performed below, the density is therefore linear in the
normal-normal component \(\nu\) of \(P^{\mu\nu}\) and at most quadratic in
the mixed components \(v^i\).  Linearity in \(\nu\) produces the scalar
constraint \(\mathcal C_0\); the quadratic dependence on \(v^i\) makes the
mixed component algebraically removable wherever the spatial auxiliary block
is invertible.  For a generic auxiliary potential the normal component would
not stay linear, and its field equation would determine it instead of
imposing a constraint.  Linearity alone does not settle the count, however:
one still has to show that preserving \(\mathcal C_0\) produces an
independent secondary constraint and that the pair is second class.  Those
are the nonlinear steps of Sec.~\ref{sec:nonlinear-hamiltonian-analysis}.

We need the maximally symmetric vacuum only through a few relations.  On
such a background,
\begin{equation}
        \bar R_{\mu\nu}=2\Lambda \bar g_{\mu\nu},
        \qquad
        \bar G_{\mu\nu}=-\Lambda \bar g_{\mu\nu},
\end{equation}
so that
\begin{equation}
        \bar A_{\mu\nu}
        =
        a\,\bar g_{\mu\nu},
        \qquad
        a:=1-\frac{\sigma\Lambda}{m^2}.
        \label{eq:a-def-main}
\end{equation}
The auxiliary equation gives
\begin{equation}
        \bar q_{\mu\nu}=a\,\bar g_{\mu\nu}.
        \label{eq:qbar-agbar-main}
\end{equation}
On the same-sign Lorentzian branch, the background field equation gives
\begin{equation}
        \beta=\sqrt a>0,
        \qquad
        a=\beta^2,
        \qquad
        \Lambda=\frac{m^2}{\sigma}(1-\beta^2).
        \label{eq:unique-vacuum-relation}
\end{equation}
Equivalently,
\begin{equation}
        2c\Lambda=1-a.
        \label{eq:two-c-Lambda}
\end{equation}
Since \(\beta=1-\lambda_0/2\), the vacuum relation gives
\begin{equation}
        \Lambda
        =
        \frac{m^2}{\sigma}\,
        \lambda_0
        \left(1-\frac{\lambda_0}{4}\right).
        \label{eq:Lambda-lambda0-relation}
\end{equation}
So the maximally symmetric vacuum is unique, one of the distinguishing
features of BINMG \cite{GSTBINMG,TekinAuxMetric2026}.  On the same-sign
Lorentzian branch \(\beta=\sqrt a>0\), hence \(\lambda_0<2\); we assume
this throughout.  Note that \(\beta\) is still the action parameter of
\eqref{eq:beta-def}; \(\beta=\sqrt a\) is a consequence of the vacuum
equation, not a second definition. See \cite{Gullu:2010st} for a detailed discussion of the vacuum of the theory and the constraints on the parameters. 

The linearized analysis around this vacuum, including the separation into a
nondynamical three-dimensional Einstein sector and a Fierz--Pauli massive
spin-two sector, was carried out in \cite{TekinAuxMetric2026}.  We use it
only as a check.

\subsection{Why the auxiliary formulation is needed}
\label{sec:why-auxiliary-needed}

Why not start from the determinant action itself?  It can be rewritten as
\cite{Gullu:2010st}
\begin{equation}
        I_{\rm BINMG}
        =
        -\frac{4m^2}{\kappa^2}
        \int \dd^3x\,\sqrt{-g}\,
        F(R,\mathscr K,\mathscr S),
        \label{eq:BI-F-form}
\end{equation}
where
\begin{equation}
        F(R,\mathscr K,\mathscr S)
        =
        \sqrt{
        1-\frac{\sigma}{2m^2}
        \left(
        R+\frac{\sigma}{m^2}\mathscr K
        -\frac{1}{12m^4}\mathscr S
        \right)}
        -
        \beta ,
        \label{eq:F-definition}
\end{equation}
with
\begin{equation}
        \mathscr K
        =
        R_{\mu\nu}R^{\mu\nu}
        -
        \frac12R^2,
        \qquad
        \mathscr S
        =
        8R_{\mu\nu}R^{\mu\alpha}R_{\alpha}{}^{\nu}
        -
        6RR_{\mu\nu}R^{\mu\nu}
        +
        R^3 .
        \label{eq:scriptK-scriptS-def}
\end{equation}
We have denoted the curvature invariant by \(\mathscr K\) in order not to
confuse it with the ADM extrinsic curvature \(K_{ij}\).

For the ADM metric, the determinant under the square root may be written in a
compact block form.  Let
\begin{equation}
        \alpha:=\frac{\sigma}{m^2}=2c,
        \qquad
        \mathscr G:=G_{\perp\perp},
        \qquad
        \mathcal M_i:=G_{\perp i},
        \qquad
        \mathcal E^i{}_j:=\gamma^{ik}G_{kj}.
        \label{eq:block-form-defs}
\end{equation}
Then
\begin{equation}
        \det
        \left(
        \delta^\mu{}_\nu
        +
        \alpha G^\mu{}_\nu
        \right)
        =
        \det(\delta^i{}_j+\alpha\mathcal E^i{}_j)
        \left[
        1-\alpha\mathscr G
        +
        \alpha^2
        \mathcal M_i
        (\delta+\alpha\mathcal E)^{-1\,i}{}_{j}
        \mathcal M^j
        \right].
        \label{eq:block-determinant-identity}
\end{equation}
Thus the direct ADM form of the determinant action is
\begin{equation}
        I_{\rm BINMG}
        =
        -\frac{4m^2}{\kappa^2}
        \int \dd t\,\dd^2x\,
        N\sqrt{\gamma}
        \left\{
        \sqrt{
        \det(\delta+\alpha\mathcal E)
        \left[
        1-\alpha\mathscr G
        +
        \alpha^2
        \mathcal M_i
        (\delta+\alpha\mathcal E)^{-1\,i}{}_{j}
        \mathcal M^j
        \right]}
        -
        \beta
        \right\}.
        \label{eq:BI-ADM-determinant-form}
\end{equation}
This is exact but not useful as a Hamiltonian starting point.  The ADM
projections are
\begin{equation}
        \mathscr G
        =
        \frac12
        \left(
        {}^{(2)}R+K^2-K_{ij}K^{ij}
        \right),
        \qquad
        \mathcal M_i
        =
        D_jK^j{}_i-D_iK,
        \label{eq:GC-ADM-projections}
\end{equation}
while the spatial projection \(G_{ij}\), and hence \(\mathcal E^i{}_j\),
contains
\((\partial_t-\mathcal L_{\vec N})K_{ij}\).
So the square root contains the time derivative of the extrinsic curvature,
a nonlinear function of second time derivatives of \(\gamma_{ij}\), and the
constraint structure is buried inside it.  In the auxiliary formulation the
same theory is polynomial, \eqref{eq:aux-P-action-expanded-pedagogical}, the
mixed component of \(P^{\mu\nu}\) is algebraic and the normal component is
linear.  That is why we start there.

\section{Nonlinear Hamiltonian analysis}
\label{sec:nonlinear-hamiltonian-analysis}

\subsection{Strategy}
\label{sec:nonlinear-dof-setup}

We first rewrite the auxiliary action in ADM variables, separating the
physical lapse and shift from the normal, mixed, and spatial components of
the auxiliary density, and put it into first-order form.  We then invert the
Legendre map, write down the canonical Hamiltonian with all auxiliary
variables still present, identify the diffeomorphism sector, and eliminate
\(v^i\).  What remains is the scalar Dirac chain generated by \(\nu\); its
rank fixes the count \eqref{eq:intro-dof-count}.

\subsection{ADM form of the auxiliary action}
\label{sec:ADM-auxiliary-action}

We begin with the polynomial density
\eqref{eq:aux-P-action-expanded-pedagogical}, repeated here for reference:
\begin{equation}
        \mathcal L
        =
        \frac12P^{\mu\nu}g_{\mu\nu}
        +
        cP^{\mu\nu}G_{\mu\nu}(g)
        +
        \frac12\det P
        -
        \beta\sqrt{-g},
        \qquad
        c=\frac{\sigma}{2m^2}.
        \label{eq:aux-P-density-for-ADM}
\end{equation}
The physical metric is written in the ADM form
\eqref{eq:ADMmetric-explained}.  In an adapted chart,
\begin{equation}
        n_\mu=(-N,0,0),
        \qquad
        n^\mu=\frac1N(1,-N^i),
        \qquad
        e_i{}^\mu=\delta_i^\mu ,
\end{equation}
with
\begin{equation}
        g_{\mu\nu}n^\mu n^\nu=-1,
        \qquad
        g_{\mu\nu}n^\mu e_i{}^\nu=0,
        \qquad
        g_{\mu\nu}e_i{}^\mu e_j{}^\nu=\gamma_{ij}.
\end{equation}
Equivalently,
\begin{equation}
        g^{\mu\nu}
        =
        -n^\mu n^\nu
        +
        \gamma^{ij}e_i{}^\mu e_j{}^\nu .
        \label{eq:inverse-metric-frame}
\end{equation}

Since \(P^{\mu\nu}\) is a symmetric contravariant tensor density of weight
\(+1\) \eqref{eq:P-def-pedagogical}, it can be decomposed in the same frame as
\begin{equation}
 P^{\mu\nu}
 =
 \nu n^\mu n^\nu
 +
 2\widetilde v^{\,i}n^{(\mu}e_i{}^{\nu)}
 +
 \widetilde S^{ij}e_i{}^\mu e_j{}^\nu .
 \label{eq:Pdecomp-raw}
\end{equation}
Under spatial coordinate transformations, \(\nu\) is a scalar density of
weight \(+1\), \(\widetilde v^{\,i}\) is a contravariant vector density of
weight \(+1\), and \(\widetilde S^{ij}\) is a symmetric contravariant tensor
density of weight \(+1\).

We absorb one power of the lapse into the mixed and spatial components,
\begin{equation}
        v^i:=\frac{\widetilde v^{\,i}}{N},
        \qquad
        S^{ij}:=\frac{\widetilde S^{ij}}{N}.
        \label{eq:lapse-rescaled-auxiliary-vars}
\end{equation}
which is invertible on a regular ADM chart, \(N\neq0\), and does not change
the spatial tensor character or density weight of \(v^i\) and \(S^{ij}\).
In these variables,
\begin{equation}
 P^{\mu\nu}
 =
 \nu n^\mu n^\nu
 +
 2Nv^i n^{(\mu}e_i{}^{\nu)}
 +
 NS^{ij}e_i{}^\mu e_j{}^\nu .
 \label{eq:Pdecomp-explained}
\end{equation}
We leave \(\nu\) unrescaled; the reason will be clear from
\eqref{eq:full-ADM-auxiliary-Lagrangian-separated}.

We shall use the abbreviations
\begin{equation}
        S:=\gamma_{ij}S^{ij},
        \qquad
        s:=\det S^{ij}.
        \label{eq:S-and-s-def}
\end{equation}
The field \(S^{ij}\) is a contravariant spatial tensor density of weight
\(+1\).  In two spatial dimensions its determinant \(s\) is nevertheless a
scalar.  Explicitly, under an orientation-preserving spatial coordinate
transformation with Jacobian \(J^i{}_j\) and determinant \(J\),
\begin{equation}
        S'^{ij}
        =
        J^{-1}J^i{}_kJ^j{}_lS^{kl}
        \quad\Longrightarrow\quad
        \det S'^{ij}=\det S^{ij}.
        \label{eq:s-is-scalar}
\end{equation}

We can now evaluate the three terms in
\eqref{eq:aux-P-density-for-ADM} that involve \(P^{\mu\nu}\).  First,
\begin{equation}
        P^{\mu\nu}g_{\mu\nu}
        =
        -\nu+NS .
        \label{eq:Pg-explained}
\end{equation}
Second, defining the projections of the Einstein tensor as 
\begin{equation}
        G_{\perp\perp}:=G_{\mu\nu}n^\mu n^\nu,
        \qquad
        G_{\perp i}:=G_{\mu\nu}n^\mu e_i{}^\nu,
        \qquad
        G_{ij}:=G_{\mu\nu}e_i{}^\mu e_j{}^\nu ,
\end{equation}
we have
\begin{equation}
        P^{\mu\nu}G_{\mu\nu}
        =
        \nu G_{\perp\perp}
        +
        2Nv^iG_{\perp i}
        +
        NS^{ij}G_{ij}.
        \label{eq:PG-explained}
\end{equation}
Finally, the determinant is
\begin{equation}
        \det P
        =
        s\left(
        \nu
        -
        Nv^i(S^{-1})_{ij}v^j
        \right),
        \label{eq:blockdet-coordinate-explained}
\end{equation}
where \((S^{-1})_{ij}\) denotes the matrix inverse of \(S^{ij}\).  This
requires the spatial auxiliary block to be invertible, \(s\neq0\).

Substituting these expressions into
\eqref{eq:aux-P-density-for-ADM} gives the full ADM auxiliary Lagrangian:
\begin{equation}
\begin{aligned}
 \mathcal L_{\rm ADM}
 ={}&
 \frac12(-\nu+NS)
 +
 c\left(
 \nu G_{\perp\perp}
 +
 2Nv^iG_{\perp i}
 +
 NS^{ij}G_{ij}
 \right)
 \\
 &+
 \frac{s}{2}
 \left(
 \nu
 -
 Nv^i(S^{-1})_{ij}v^j
 \right)
 -
 \beta N\sqrt\gamma .
\end{aligned}
\label{eq:full-ADM-auxiliary-Lagrangian}
\end{equation}
Equivalently, separating the coefficients of the physical lapse \(N\) and
the auxiliary normal component \(\nu\),
\begin{equation}
\boxed{
\begin{aligned}
 \mathcal L_{\rm ADM}
 ={}&
 N\left[
 \frac12S
 +
 2cv^iG_{\perp i}
 +
 cS^{ij}G_{ij}
 -
 \frac{s}{2}v^i(S^{-1})_{ij}v^j
 -
 \beta\sqrt\gamma
 \right]
 \\
 &+
 \nu\left[
 -\frac12
 +
 cG_{\perp\perp}
 +
 \frac{s}{2}
 \right].
\end{aligned}
}
\label{eq:full-ADM-auxiliary-Lagrangian-separated}
\end{equation}
This is the starting point of the Hamiltonian analysis.  The lapse \(N\) and
the auxiliary normal variable \(\nu\) are distinct fields; the Lagrangian is
linear in \(\nu\) and at most quadratic in \(v^i\), as anticipated in
Sec.~\ref{sec:BINMG}.

We next express the Einstein-tensor projections in terms of the ADM
variables.  With the extrinsic-curvature convention
\eqref{eq:K-convention-early}, the normal-normal and normal-spatial projections are \cite{non}
\begin{align}
 G_{\perp\perp}
 &=
 \frac12
 \left(
 {}^{(2)}R+K^2-K_{ij}K^{ij}
 \right),
 \label{eq:Gpp-explained}
 \\
 G_{\perp i}
 &=
 D_jK^j{}_i-D_iK,
 \label{eq:Gpi-explained}
\end{align}
which do not contain a time derivative of \(K_{ij}\).  With the normal-evolution derivative
\begin{equation}
 \mathscr D_t
 :=
 \partial_t-\mathcal L_{\vec N}.
 \label{eq:Dt-def}
\end{equation}
the spatial projection of the Einstein tensor in \(2+1\) dimensions is
\begin{equation}
\begin{aligned}
 G_{ij}
 ={}&
 \frac1N
 \left[
 \mathscr D_tK_{ij}
 -
 \gamma_{ij}\gamma^{kl}\mathscr D_tK_{kl}
 -
 D_iD_jN
 +
 \gamma_{ij}D^2N
 \right]
 \\
 &+
 KK_{ij}
 -
 2K_i{}^kK_{kj}
 -
 \frac12\gamma_{ij}K^2
 +
 \frac32\gamma_{ij}K_{kl}K^{kl}.
\end{aligned}
\label{eq:Gij-velocity-split}
\end{equation}

The coefficient \(3/2\) in the last term follows from the fact that the
trace \(K=\gamma^{ij}K_{ij}\) itself contains the evolving inverse spatial
metric.  Indeed,
\begin{equation}
 \mathscr D_t\gamma_{ij}=2NK_{ij},
 \qquad
 \mathscr D_t\gamma^{ij}=-2NK^{ij},
\end{equation}
and hence
\begin{equation}
 \mathscr D_tK
 =
 \gamma^{ij}\mathscr D_tK_{ij}
 -
 2N K_{ij}K^{ij}.
 \label{eq:DtK-trace-identity}
\end{equation}
Using this identity, \eqref{eq:Gij-velocity-split} may equivalently be
written as
\begin{equation}
\begin{aligned}
 G_{ij}
 ={}&
 \frac1N
 \left[
 \mathscr D_tK_{ij}
 -
 \gamma_{ij}\mathscr D_tK
 -
 D_iD_jN
 +
 \gamma_{ij}D^2N
 \right]
 \\
 &+
 KK_{ij}
 -
 2K_i{}^kK_{kj}
 -
 \frac12\gamma_{ij}
 \left(
 K^2+K_{kl}K^{kl}
 \right).
\end{aligned}
\label{eq:Gij-exact-geometric}
\end{equation}
There is no intrinsic Einstein-tensor term in the spatial projection because, in two dimensions, the Einstein tensor vanishes identically: 
\begin{equation}
 {}^{(2)}G_{ij}
 =
 {}^{(2)}R_{ij}
 -
 \frac12\gamma_{ij}\,{}^{(2)}R
 \equiv0.
\end{equation}
The intrinsic curvature therefore drops out of \(G_{ij}\); it survives in
\eqref{eq:Gpp-explained} and will enter the scalar constraint.
 
We now contract the spatial projection with \(S^{ij}\).  Define the
trace-adjusted spatial auxiliary variable, which contains the same information  as \(S^{ij}\), as  
\begin{equation}
 \mathcal A^{ij}
 :=
 S^{ij}-S\gamma^{ij},
 \qquad
 S:=\gamma_{ij}S^{ij}.
 \label{eq:Aij-def-explained}
\end{equation}
Then direct substitution of \eqref{eq:Gij-velocity-split} gives
\begin{equation}
\begin{aligned}
 NS^{ij}G_{ij}
 ={}&
 \mathcal A^{ij}\mathscr D_tK_{ij}
 -
 \mathcal A^{ij}D_iD_jN
 \\
 &+
 N\Bigg[
 KS^{ij}K_{ij}
 -
 2S^{ij}K_i{}^kK_{kj}
 -
 \frac S2K^2
 +
 \frac{3S}{2}K_{ij}K^{ij}
 \Bigg].
\end{aligned}
\label{eq:SijGij-Aij-split}
\end{equation}
The factor \(N\) multiplying
\(S^{ij}G_{ij}\) in
\eqref{eq:full-ADM-auxiliary-Lagrangian-separated} cancels the \(1/N\)
multiplying the normal time derivative of \(K_{ij}\).

For later use, define
\begin{equation}
 \mathcal Q_K(S,K)
 :=
 KS^{ij}K_{ij}
 -
 2S^{ij}K_i{}^kK_{kj}
 -
 \frac S2K^2
 +
 \frac{3S}{2}K_{ij}K^{ij}.
 \label{eq:QK-def}
\end{equation}
Substituting
\eqref{eq:Gpp-explained},
\eqref{eq:Gpi-explained}, and
\eqref{eq:SijGij-Aij-split}
into \eqref{eq:full-ADM-auxiliary-Lagrangian-separated} gives the complete
ADM Lagrangian
\begin{equation}
\begin{aligned}
 \mathcal L_{\rm ADM}
 ={}&
 c\mathcal A^{ij}\dot K_{ij}
 -
 c\mathcal A^{ij}\mathcal L_{\vec N}K_{ij}
 -
 c\mathcal A^{ij}D_iD_jN
 \\
 &+
 N\Bigg[
 \frac12S
 -
 \beta\sqrt\gamma
 +
 2cv^i\left(D_jK^j{}_i-D_iK\right)
 -
 \frac{s}{2}v^i(S^{-1})_{ij}v^j
 +
 c\mathcal Q_K(S,K)
 \Bigg]
 \\
 &+
 \nu\Bigg[
 \frac{s-1}{2}
 +
 \frac c2\,{}^{(2)}R
 +
 \frac c2
 \left(
 K^2-K_{ij}K^{ij}
 \right)
 \Bigg].
\end{aligned}
\label{eq:full-ADM-Lagrangian-exact}
\end{equation}
We now remove the time derivative of the extrinsic curvature, which
involves \(\ddot\gamma_{ij}\).  The first term of
\eqref{eq:full-ADM-Lagrangian-exact} is its only occurrence.  Integrating by
parts in time,
\begin{equation}
 c\mathcal A^{ij}\dot K_{ij}
 =
 \partial_t\left(c\mathcal A^{ij}K_{ij}\right)
 -
 c\dot{\mathcal A}^{ij}K_{ij}.
\end{equation}
Using the convention \(\doteq\) introduced in
Sec.~\ref{sec:conventions} for equality modulo boundary terms,
\begin{equation}
 c\mathcal A^{ij}\dot K_{ij}
 \doteq
 -
 c\dot{\mathcal A}^{ij}K_{ij}.
 \label{eq:time-IBP-AK}
\end{equation}
We also remove derivatives of the physical shift from the explicit
Lie-derivative term.  Since
\begin{equation}
 (\mathcal L_{\vec N}K)_{ij}
 =
 N^kD_kK_{ij}
 +
 K_{kj}D_iN^k
 +
 K_{ik}D_jN^k,
\end{equation}
symmetry of \(\mathcal A^{ij}\) and \(K_{ij}\) gives
\begin{align}
 -
 \mathcal A^{ij}\mathcal L_{\vec N}K_{ij}
 &=
 -
 N^k\mathcal A^{ij}D_kK_{ij}
 -
 2\mathcal A^{ij}K_{kj}D_iN^k
 \nonumber\\
 &\doteq
 N^k
 \left[
 2D_i\!\left(\mathcal A^{ij}K_{kj}\right)
 -
 \mathcal A^{ij}D_kK_{ij}
 \right].
\end{align}
We therefore define
\begin{equation}
 \mathcal J_k
 :=
 2D_i\!\left(\mathcal A^{ij}K_{kj}\right)
 -
 \mathcal A^{ij}D_kK_{ij}.
 \label{eq:Jk-def}
\end{equation}

Finally, the term containing second spatial derivatives of the lapse may be
integrated twice by parts:
\begin{align}
 -
 \mathcal A^{ij}D_iD_jN
 &\doteq
 (D_i\mathcal A^{ij})D_jN
 \doteq
 -
 N D_jD_i\mathcal A^{ij}.
\end{align}
The resulting first-order bulk Lagrangian is
\begin{equation}
\begin{aligned}
 \mathcal L_{\rm 1st}
 ={}&
 -
 c\dot{\mathcal A}^{ij}K_{ij}
 +
 cN^i\mathcal J_i
 \\
 &+
 N\Bigg[
 \frac12S
 -
 \beta\sqrt\gamma
 +
 2cv^i\left(D_jK^j{}_i-D_iK\right)
 -
 \frac{s}{2}v^i(S^{-1})_{ij}v^j
 \\
 &\hspace{19mm}
 +
 c\mathcal Q_K(S,K)
 -
 cD_iD_j\mathcal A^{ij}
 \Bigg]
 \\
 &+
 \nu\Bigg[
 \frac{s-1}{2}
 +
 \frac c2\,{}^{(2)}R
 +
 \frac c2
 \left(
 K^2-K_{ij}K^{ij}
 \right)
 \Bigg].
\end{aligned}
\label{eq:first-order-Lagrangian-exact}
\end{equation}
``First order'' means that no \(\dot K_{ij}\), hence no second time
derivative of the spatial metric, remains.  The two densities differ by
boundary terms only:
\begin{equation}
 \mathcal L_{\rm ADM}
 =
 \mathcal L_{\rm 1st}
 +
 \partial_t\left(c\mathcal A^{ij}K_{ij}\right)
 +
 D_i\mathcal B^i,
 \label{eq:first-order-boundary-relation}
\end{equation}
where
\begin{equation}
 \mathcal B^i
 =
 c\left[
 N D_j\mathcal A^{ij}
 -
 \mathcal A^{ij}D_jN
 -
 2N^k\mathcal A^{ij}K_{kj}
 \right].
 \label{eq:spatial-boundary-density}
\end{equation}
The velocity structure is now visible.  The six dynamical velocities are
those of \(\gamma_{ij}\) and \(\mathcal A^{ij}\):
\[
 \dot\gamma_{ij},
 \qquad
 \dot{\mathcal A}^{ij}.
\]
None of \(\dot N\), \(\dot N^i\), \(\dot\nu\), \(\dot v^i\) occurs, so the
Legendre map lives entirely in this six-dimensional block.

The mixed kinetic term takes the particularly simple form
\begin{equation}
 -
 c\dot{\mathcal A}^{ij}K_{ij}
 =
 -
 \frac{c}{2N}
 \dot{\mathcal A}^{ij}
 \left(
 \dot\gamma_{ij}-D_iN_j-D_jN_i
 \right).
 \label{eq:mixed-kinetic-exact}
\end{equation}
The momentum conjugate to the temporary configuration variable
\(\mathcal A^{ij}\) is
\begin{equation}
 \Pi_{ij}
 :=
 \frac{\partial\mathcal L_{\rm 1st}}
 {\partial\dot{\mathcal A}^{ij}} .
 \label{eq:Pi-def}
\end{equation}
It follows immediately that
\begin{equation}
 \Pi_{ij}=-cK_{ij}.
 \label{eq:Pi-minus-cK-preview}
\end{equation}
Hence, on a regular ADM chart,
\begin{equation}
 \dot\gamma_{ij}
 =
 -\frac{2N}{c}\Pi_{ij}
 +
 D_iN_j+D_jN_i .
 \label{eq:dotgamma-from-Pi}
\end{equation}
So \(\Pi_{ij}\) determines \(\dot\gamma_{ij}\).  The remaining three
velocities, in \(\dot{\mathcal A}^{ij}\), are determined by the momenta
conjugate to \(\gamma_{ij}\).

\subsection{Velocity Hessian, Legendre map, and canonical variables}
\label{sec:velocity-block-explained}

Before differentiating with respect to the velocities we move all spatial
derivatives off \(K_{ij}\), so that the dependence on \(\dot\gamma_{ij}\)
becomes algebraic.  For the shift term in
\eqref{eq:first-order-Lagrangian-exact}, up to boundary terms,
\begin{equation}
        N^i\mathcal J_i
        \doteq
        K_{ij}(\mathcal L_{\vec N}\mathcal A)^{ij},
        \label{eq:shift-J-to-LieA}
\end{equation}
where, because \(\mathcal A^{ij}\) is a contravariant tensor density of
weight \(+1\),
\begin{equation}
\begin{aligned}
        (\mathcal L_{\vec N}\mathcal A)^{ij}
        ={}&
        N^kD_k\mathcal A^{ij}
        -
        \mathcal A^{kj}D_kN^i
        -
        \mathcal A^{ik}D_kN^j
        +
        \mathcal A^{ij}D_kN^k .
\end{aligned}
        \label{eq:Lie-A-density}
\end{equation}
Similarly, for the \(v^i\) term,
\begin{align}
        2Nv^i
        \left(
        D_jK^j{}_i-D_iK
        \right)
        &\doteq
        2K_{ij}
        \left[
        \gamma^{ij}D_k(Nv^k)
        -
        D^{(i}(Nv^{j)})
        \right].
        \label{eq:vK-integration-by-parts}
\end{align}
The first-order Lagrangian then reads
\begin{equation}
\boxed{
\begin{aligned}
 \mathcal L_{\rm Leg}
 ={}&
 -cK_{ij}
 \left[
 \dot{\mathcal A}^{ij}
 -
 (\mathcal L_{\vec N}\mathcal A)^{ij}
 \right]
 \\
 &+
 2cK_{ij}
 \left[
 \gamma^{ij}D_k(Nv^k)
 -
 D^{(i}(Nv^{j)})
 \right]
 \\
 &+
 Nc\,\mathcal Q_K(S,K)
 +
 \frac{c\nu}{2}
 \left(
 K^2-K_{ij}K^{ij}
 \right)
 \\
 &+
 N\left[
 \frac12S
 -
 \beta\sqrt\gamma
 -
 \frac{s}{2}v^i(S^{-1})_{ij}v^j
 -
 cD_iD_j\mathcal A^{ij}
 \right]
 \\
 &+
 \nu\left[
 \frac{s-1}{2}
 +
 \frac c2\,{}^{(2)}R
 \right].
\end{aligned}}
\label{eq:Legendre-ready-Lagrangian}
\end{equation}
with \(\mathcal Q_K\) from \eqref{eq:QK-def}.

The momentum conjugate to \(\mathcal A^{ij}\) was obtained in
Eq.~\eqref{eq:Pi-minus-cK-preview} and Eq.~\eqref{eq:dotgamma-from-Pi} gives \(\dot\gamma_{ij}\) in terms of
\(\Pi_{ij}\), the lapse, and the shift.  The momentum conjugate to
\(\gamma_{ij}\), at fixed \(\mathcal A^{ij}\), is
\begin{equation}
        \pi^{ij}
        :=
        \left.
        \frac{\partial\mathcal L_{\rm Leg}}
        {\partial\dot\gamma_{ij}}
        \right|_{\mathcal A}.
        \label{eq:pi-temporary-def}
\end{equation}
Since
\begin{equation}
        \frac{\partial K_{kl}}
        {\partial\dot\gamma_{ij}}
        =
        \frac1{2N}
        \delta_{(k}{}^i\delta_{l)}{}^j,
        \label{eq:dK-ddotgamma}
\end{equation}
we may equivalently write
\begin{equation}
        \pi^{ij}
        =
        \frac1{2N}
        \frac{\partial\mathcal L_{\rm Leg}}
        {\partial K_{ij}} .
        \label{eq:pi-from-K-derivative}
\end{equation}
The derivative of the quadratic expression \(\mathcal Q_K\) is
\begin{equation}
\begin{aligned}
 \frac{\partial\mathcal Q_K}{\partial K_{ij}}
 ={}&
 \left(
 S^{ab}K_{ab}-SK
 \right)\gamma^{ij}
 +
 KS^{ij}
 -
 4S^{k(i}K_k{}^{j)}
 +
 3SK^{ij}.
\end{aligned}
\label{eq:dQK-dK}
\end{equation}
Then we have 
\begin{equation}
\begin{aligned}
 \pi^{ij}
 ={}&
 -\frac{c}{2N}
 \left[
 \dot{\mathcal A}^{ij}
 -
 (\mathcal L_{\vec N}\mathcal A)^{ij}
 \right]
 \\
 &+
 \frac{c}{N}
 \left[
 \gamma^{ij}D_k(Nv^k)
 -
 D^{(i}(Nv^{j)})
 \right]
 \\
 &+
 \frac c2
 \left[
 \left(
 S^{ab}K_{ab}-SK
 \right)\gamma^{ij}
 +
 KS^{ij}
 -
 4S^{k(i}K_k{}^{j)}
 +
 3SK^{ij}
 \right]
 \\
 &+
 \frac{c\nu}{2N}
 \left(
 K\gamma^{ij}-K^{ij}
 \right).
\end{aligned}
\label{eq:pi-explained-correct}
\end{equation}
Solving this for the remaining three velocities gives
\begin{equation}
\begin{aligned}
 \dot{\mathcal A}^{ij}
 ={}&
 (\mathcal L_{\vec N}\mathcal A)^{ij}
 +
 2\left[
 \gamma^{ij}D_k(Nv^k)
 -
 D^{(i}(Nv^{j)})
 \right]
 -
 \frac{2N}{c}\pi^{ij}
 \\
 &+
 N\left[
 \left(
 S^{ab}K_{ab}-SK
 \right)\gamma^{ij}
 +
 KS^{ij}
 -
 4S^{k(i}K_k{}^{j)}
 +
 3SK^{ij}
 \right]
 \\
 &+
 \nu
 \left(
 K\gamma^{ij}-K^{ij}
 \right).
\end{aligned}
\label{eq:Adot-solved-exact}
\end{equation}
Using \(K_{ij}=-\Pi_{ij}/c\), this becomes
\begin{equation}
\begin{aligned}
 \dot{\mathcal A}^{ij}
 ={}&
 (\mathcal L_{\vec N}\mathcal A)^{ij}
 +
 2\left[
 \gamma^{ij}D_k(Nv^k)
 -
 D^{(i}(Nv^{j)})
 \right]
 -
 \frac{2N}{c}\pi^{ij}
 \\
 &+
 \frac{N}{c}
 \Big[
 (S\Pi-S^{ab}\Pi_{ab})\gamma^{ij}
 -
 \Pi S^{ij}
 +
 4S^{k(i}\Pi_k{}^{j)}
 -
 3S\Pi^{ij}
 \Big]
 \\
 &-
 \frac{\nu}{c}
 \left(
 \Pi\gamma^{ij}-\Pi^{ij}
 \right),
 \qquad
 \Pi:=\gamma^{ij}\Pi_{ij}.
\end{aligned}
\label{eq:Adot-solved-Pi-exact}
\end{equation}
Equations \eqref{eq:dotgamma-from-Pi} and
\eqref{eq:Adot-solved-Pi-exact} invert the Legendre map for all six
velocities.  The velocity Hessian makes the absence of primary constraints
in this sector clear .  At each spatial point, both
\(\dot\gamma_{ij}\) and \(\dot{\mathcal A}^{ij}\) have three independent
components.  Choose dual bases \(E^A{}_{ij}\) and \(E_A{}^{ij}\),
\(A=1,2,3\), normalized by
\begin{equation}
        E_A{}^{ij}E^B{}_{ij}
        =
        \delta_A{}^B .
        \label{eq:symmetric-tensor-basis}
\end{equation}
Write
\begin{equation}
        \dot\gamma_{ij}
        =
        u^A E^A{}_{ij},
        \qquad
        \dot{\mathcal A}^{ij}
        =
        w^A E_A{}^{ij}.
        \label{eq:u-w-def}
\end{equation}
Then
\begin{equation}
        \dot{\mathcal A}^{ij}\dot\gamma_{ij}
        =
        w^Au^A .
        \label{eq:u-w-contraction}
\end{equation}

Let
\begin{equation}
        z^\alpha
        :=
        (u^1,u^2,u^3,w^1,w^2,w^3),
        \qquad
        \mathbb W_{\alpha\beta}
        :=
        \frac{\partial^2\mathcal L_{\rm Leg}}
        {\partial z^\alpha\partial z^\beta}.
        \label{eq:velocity-Hessian-def}
\end{equation}
The mixed kinetic term contains
\begin{equation}
        -c\dot{\mathcal A}^{ij}K_{ij}
        =
        -\frac{c}{2N}w^Au^A
        +\text{terms linear in }w^A ,
        \label{eq:mixed-u-w}
\end{equation}
and there is no term quadratic in \(w^A\).  Therefore the exact Hessian has
the block form
\begin{equation}
        \mathbb W
        =
        \begin{pmatrix}
        M & -\dfrac{c}{2N}\mathbf 1_3
        \\[2mm]
        -\dfrac{c}{2N}\mathbf 1_3 & 0
        \end{pmatrix},
        \label{eq:velocity-Hessian-block}
\end{equation}
where
\begin{equation}
        M_{AB}
        :=
        \frac{\partial^2\mathcal L_{\rm Leg}}
        {\partial u^A\partial u^B}.
        \label{eq:M-Hessian-def}
\end{equation}
Only terms quadratic in \(K_{ij}\) contribute to this block, and explicitly
\begin{equation}
\begin{aligned}
 M_{AB}
 ={}&
 \frac{c}{4N}
 E^A{}_{ij}
 \frac{\partial^2\mathcal Q_K}
      {\partial K_{ij}\partial K_{kl}}
 E^B{}_{kl}+
 \frac{c\nu}{4N^2}
 E^A{}_{ij}
 \left(
 \gamma^{ij}\gamma^{kl}
 -
 \gamma^{i(k}\gamma^{l)j}
 \right)
 E^B{}_{kl}.
\end{aligned}
\label{eq:M-Hessian-explicit}
\end{equation}
Its entries do not matter, because the off-diagonal block of
\eqref{eq:velocity-Hessian-block} is nonsingular:
\begin{equation}
        \mathbb W^{-1}
        =
        \begin{pmatrix}
        0 & -\dfrac{2N}{c}\mathbf 1_3
        \\[2mm]
        -\dfrac{2N}{c}\mathbf 1_3
        &
        -\dfrac{4N^2}{c^2}M
        \end{pmatrix}.
        \label{eq:velocity-Hessian-inverse}
\end{equation}
Hence, on a regular ADM chart, the Hessian is invertible whenever
\begin{equation}
        c\neq0 .
        \label{eq:c-nonzero-invertibility}
\end{equation}
For BINMG,
\begin{equation}
        c=\frac{\sigma}{2m^2},
\end{equation}
so this holds for finite \(m^2\) and \(\sigma=\pm1\).  There is no primary
constraint in the sector \((\dot\gamma_{ij},\dot{\mathcal A}^{ij})\).  (The
condition \(N\neq0\) only says that the ADM chart is regular; it is not a
branch condition on the phase space.)

We now return from \(\mathcal A^{ij}\) to \(S^{ij}\).  The momenta
transform most directly through the symplectic one-form,
\begin{equation}
        \Theta
        =
        \pi^{ij}\dd\gamma_{ij}
        +
        \Pi_{ij}\dd\mathcal A^{ij}.
        \label{eq:symplectic-temporary}
\end{equation}
Using
\begin{equation}
        \mathcal A^{ij}
        =
        S^{ij}-S\gamma^{ij},
        \qquad
        S=\gamma_{ab}S^{ab},
        \label{eq:A-S-relation-canonical}
\end{equation}
we have
\begin{equation}
        \dd S
        =
        S^{ab}\dd\gamma_{ab}
        +
        \gamma_{ab}\dd S^{ab},
        \label{eq:dS-canonical}
\end{equation}
and
\begin{equation}
        \dd\gamma^{ij}
        =
        -\gamma^{ia}\gamma^{jb}\dd\gamma_{ab}.
        \label{eq:dgammainverse-canonical}
\end{equation}
It follows that
\begin{align}
        \dd\mathcal A^{ij} =
        \dd S^{ij}
        -
        \gamma^{ij}
        \left(
        S^{ab}\dd\gamma_{ab}
        +
        \gamma_{ab}\dd S^{ab}
        \right)+
        S\gamma^{ia}\gamma^{jb}\dd\gamma_{ab}.
        \label{eq:dA-canonical}
\end{align}

Substituting \eqref{eq:dA-canonical} into
\eqref{eq:symplectic-temporary} and collecting separately the coefficients
of \(\dd\gamma_{ij}\) and \(\dd S^{ij}\), we obtain
\begin{equation}
        \Theta
        =
        \rho^{ij}\dd\gamma_{ij}
        +
        p_{ij}\dd S^{ij},
        \label{eq:symplectic-final}
\end{equation}
where
\begin{equation}
        \rho^{ij}
        =
        \pi^{ij}
        -
        \Pi S^{ij}
        +
        S\Pi^{ij},
        \label{eq:rho-def-explained}
\end{equation}
and
\begin{equation}
        p_{ij}
        =
        \Pi_{ij}
        -
        \Pi\gamma_{ij}.
        \label{eq:p-def-explained}
\end{equation}
Here
\begin{equation}
        \Pi:=\gamma^{ij}\Pi_{ij},
        \qquad
        \Pi^{ij}:=\gamma^{ik}\gamma^{jl}\Pi_{kl}.
\end{equation}

The final dynamical canonical pairs are therefore
\begin{equation}
        (\gamma_{ij},\rho^{ij}),
        \qquad
        (S^{ij},p_{ij}).
        \label{eq:canonical-pairs-final}
\end{equation}
Their nonzero elementary Poisson brackets are
\begin{equation}
        \{\gamma_{ij}(x),\rho^{kl}(y)\}
        =
        \delta_{(i}{}^k\delta_{j)}{}^l
        \delta^{(2)}(x-y),
        \label{eq:gamma-rho-bracket}
\end{equation}
and
\begin{equation}
        \{S^{ij}(x),p_{kl}(y)\}
        =
        \delta_{(k}{}^i\delta_{l)}{}^j
        \delta^{(2)}(x-y).
        \label{eq:S-p-bracket}
\end{equation}

The relations between the temporary and final momenta simplify further in
two spatial dimensions.  Taking the trace of
\eqref{eq:p-def-explained} gives
\begin{equation}
        p
        :=
        \gamma^{ij}p_{ij}
        =
        \Pi-2\Pi
        =
        -\Pi .
        \label{eq:p-trace-identity}
\end{equation}
Hence
\begin{equation}
        \Pi=-p,
        \qquad
        \Pi_{ij}
        =
        p_{ij}-p\gamma_{ij},
        \qquad
        p:=\gamma^{ij}p_{ij}.
        \label{eq:Pi-in-terms-of-p-explained}
\end{equation}
Raising the indices gives
\begin{equation}
        \Pi^{ij}
        =
        p^{ij}-p\gamma^{ij},
        \qquad
        p^{ij}:=\gamma^{ik}\gamma^{jl}p_{kl}.
        \label{eq:Pi-upper-in-terms-of-p}
\end{equation}
Using these relations in \eqref{eq:rho-def-explained}, the temporary metric
momentum becomes
\begin{equation}
        \pi^{ij}
        =
        \rho^{ij}
        -
        pS^{ij}
        -
        Sp^{ij}
        +
        Sp\gamma^{ij}.
        \label{eq:pi-in-terms-of-rho}
\end{equation}

The transformation
\[
(\gamma_{ij},\mathcal A^{ij};\pi^{ij},\Pi_{ij})
\longrightarrow
(\gamma_{ij},S^{ij};\rho^{ij},p_{ij})
\]
is canonical by construction.

The variables \(N\), \(N^i\), \(\nu\), and \(v^i\) have no velocities, but
their roles differ: \(N\) and \(N^i\) are the ADM gauge variables, \(v^i\)
will be determined algebraically where the spatial auxiliary block is
invertible, and \(\nu\), which appears linearly, will impose a scalar
constraint.

\subsection{Canonical Hamiltonian before auxiliary reduction}
\label{sec:canonical-Hamiltonian-before-reduction}

Up to the common overall normalization of the auxiliary action, the
canonical Hamiltonian is obtained from
\begin{equation}
        \mathcal H_{\rm can}
        =
        \pi^{ij}\dot\gamma_{ij}
        +
        \Pi_{ij}\dot{\mathcal A}^{ij}
        -
        \mathcal L_{\rm Leg}.
        \label{eq:canonical-Hamiltonian-definition}
\end{equation}
Because
\(\Pi_{ij}=-cK_{ij}\), the term
\(\Pi_{ij}\dot{\mathcal A}^{ij}\) cancels the corresponding
\(-cK_{ij}\dot{\mathcal A}^{ij}\) term in
\(\mathcal L_{\rm Leg}\).  Using
\eqref{eq:dotgamma-from-Pi}, expressing the result in the final canonical
variables, and performing only the spatial integrations by parts already
allowed in the local analysis, one obtains
\begin{equation}
 I_{\rm can}
 =
 \int\dd t\,\dd^2x
 \left[
 \rho^{ij}\dot\gamma_{ij}
 +
 p_{ij}\dot S^{ij}
 -
 N\mathcal H_N
 -
 N^i\mathcal H_i
 -
 \nu\mathcal C_0
 \right].
        \label{eq:canonical-action-before-v-elimination}
\end{equation}
Here \(v^i\) has not yet been eliminated; all of its dependence sits in the
lapse coefficient \(\mathcal H_N\).  Recall
\begin{equation}
        \Pi_{ij}
        =
        p_{ij}-p\gamma_{ij},
        \qquad
        \Pi=-p,
        \label{eq:Pi-recall-full-Hamiltonian}
\end{equation}
and
\begin{equation}
        \pi^{ij}
        =
        \rho^{ij}
        -
        pS^{ij}
        -
        Sp^{ij}
        +
        Sp\gamma^{ij}.
        \label{eq:pi-recall-full-Hamiltonian}
\end{equation}
We shall also use
\begin{equation}
        \mathcal A^{ij}
        =
        S^{ij}-S\gamma^{ij},
        \qquad
        B_i:=D_jp^j{}_i,
        \qquad
        s:=\det S^{ij}.
        \label{eq:B-definition-full-Hamiltonian}
\end{equation}
The quadratic expression inherited from \(\mathcal Q_K(S,K)\) is
\begin{equation}
\begin{aligned}
        \mathcal Q(S,\Pi)
        :={}&
        \frac{3S}{2}\Pi_{ij}\Pi^{ij}
        -
        \frac S2\Pi^2
        +
        \Pi S^{ij}\Pi_{ij}
        -
        2S^{ij}\Pi_i{}^k\Pi_{kj}.
\end{aligned}
        \label{eq:Q-definition-full-Hamiltonian}
\end{equation}
Indeed,
\begin{equation}
        K_{ij}=-\frac1c\Pi_{ij}
        \qquad\Longrightarrow\qquad
        c\,\mathcal Q_K(S,K)
        =
        \frac1c\,\mathcal Q(S,\Pi).
        \label{eq:QK-Q-full-Hamiltonian}
\end{equation}

The coefficient of the lapse is
\begin{equation}
\begin{aligned}
        \mathcal H_N
        ={}&
        \beta\sqrt\gamma
        -
        \frac12S
        -
        \frac2c\pi^{ij}\Pi_{ij}
        -
        \frac1c\mathcal Q(S,\Pi)
        +
        cD_iD_j\mathcal A^{ij}
        \\
        &+
        2v^iB_i
        +
        \frac{s}{2}
        (S^{-1})_{ij}v^iv^j .
\end{aligned}
        \label{eq:HN-before-v-elimination}
\end{equation}
The coefficient of the shift is
\begin{equation}
\begin{aligned}
        \mathcal H_i
        ={}&
        -2\gamma_{ij}D_k\rho^{jk}
        +
        p_{kl}D_iS^{kl}
        +
        2D_k(p_{ij}S^{kj})
        -
        D_i(p_{kl}S^{kl}) .
\end{aligned}
        \label{eq:spatial-generator-explicit}
\end{equation}
The next subsection shows that its smeared form is the canonical generator
of spatial diffeomorphisms.  Finally, the coefficient of \(\nu\) is
\begin{equation}
        \mathcal C_0
        =
        \frac12(1-s)
        -
        \frac c2\,{}^{(2)}R
        -
        \frac1{2c}
        \left(
        p^2-p_{ij}p^{ij}
        \right).
        \label{eq:C0-before-v-elimination}
\end{equation}
Since \(\mathcal C_0\) does not contain \(v^i\), eliminating \(v^i\) will
change \(\mathcal H_N\) but not \(\mathcal C_0\).

\subsection{Diffeomorphism generators}
\label{sec:diffeomorphism-generators}

We now identify the canonical generators associated with spacetime
diffeomorphisms.  There are two logically distinct steps.  First,
spacetime covariance determines the number and structure of the gauge
chains.  Second, we identify explicitly the secondary generators that
appear in those chains.  The tangential generators can be obtained
directly on the canonical phase space, whereas the normal generator
requires a projection away from the auxiliary second-class sector.

\paragraph{Spacetime diffeomorphisms and the gauge chains.}

The auxiliary action is invariant under spacetime diffeomorphisms.  Under
an infinitesimal diffeomorphism generated by a spacetime vector field
\(\xi^\mu\),
\begin{equation}
        \delta_\xi g_{\mu\nu}
        =
        \mathcal L_\xi g_{\mu\nu},
        \qquad
        \delta_\xi P^{\mu\nu}
        =
        \mathcal L_\xi P^{\mu\nu}.
\end{equation}
The second Lie derivative is that of a contravariant tensor density of
weight \(+1\).  Since the auxiliary Lagrangian is itself a scalar density
of weight \(+1\),
\begin{equation}
        \delta_\xi\mathcal L_{\rm aux}
        =
        \partial_\mu
        \left(
        \xi^\mu\mathcal L_{\rm aux}
        \right),
\end{equation}
so the action is invariant modulo the usual boundary term.

Relative to the ADM foliation, decompose the spacetime vector field as
\begin{equation}
        \xi^\mu
        =
        \epsilon^\perp n^\mu
        +
        \epsilon^i e_i{}^\mu ,
        \label{eq:xi-normal-tangential}
\end{equation}
where
\[
        \epsilon^\perp,
        \qquad
        \epsilon^1,
        \qquad
        \epsilon^2
\]
are three independent gauge functions.  In coordinate components,
\begin{equation}
        \xi^0
        =
        \frac{\epsilon^\perp}{N},
        \qquad
        \xi^i
        =
        \epsilon^i
        -
        \frac{N^i}{N}\epsilon^\perp .
        \label{eq:descriptor-coordinate-relation}
\end{equation}
Substituting this decomposition into the Lie derivative of the metric gives
the standard ADM transformations of the lapse and shift,
\begin{align}
 \delta_\epsilon N
 &=
 \dot\epsilon^\perp
 +
 \epsilon^jD_jN
 -
 N^jD_j\epsilon^\perp ,
 \label{eq:lapse-gauge-transformation}
 \\
 \delta_\epsilon N^i
 &=
 \dot\epsilon^i
 +
 \epsilon^jD_jN^i
 -
 N^jD_j\epsilon^i
 +
 \gamma^{ij}
 \left(
 \epsilon^\perp D_jN
 -
 ND_j\epsilon^\perp
 \right).
 \label{eq:shift-gauge-transformation}
\end{align}
The variations of the remaining fields may likewise be obtained from their
Lie derivatives, but their explicit form is not needed at this point. What matters for the structure of the Dirac gauge generator is the
dependence on time derivatives of the three gauge functions.  The
coefficients of
\[
        \dot\epsilon^\perp,
        \qquad
        \dot\epsilon^1,
        \qquad
        \dot\epsilon^2
\]
in
\[
        \delta N,
        \qquad
        \delta N^1,
        \qquad
        \delta N^2
\]
form the \(3\times3\) identity matrix. Thus, all three gauge functions
enter independently through their first time derivatives, while no second
time derivatives occur. Accordingly, the Castellani construction
\cite{Castellani,HenneauxTeitelboim} organizes spacetime
diffeomorphisms into three two-stage gauge chains,
\begin{equation}
        \Phi_\alpha
        \ \longrightarrow\
        \mathcal G_\alpha,
        \qquad
        \alpha=\perp,1,2 ,
        \label{eq:diffeomorphism-two-stage-chains}
\end{equation}
where \(\Phi_\alpha\) denotes an appropriate primary first-class
representative in the lapse--shift sector and
\(\mathcal G_\alpha\) the corresponding secondary generator.

The full infinitesimal gauge transformation is therefore generated
canonically by a functional of the schematic form
\begin{equation}
        G[\epsilon]
        =
        \int_{\Sigma_t}\dd^2x
        \left[
        \dot\epsilon^\alpha\Phi_\alpha
        +
        \epsilon^\alpha
        \left(
        \mathcal G_\alpha
        +
        C_\alpha{}^\beta\Phi_\beta
        \right)
        \right].
        \label{eq:Castellani-generator-schematic}
\end{equation}
For any canonical variable \(F\), its infinitesimal diffeomorphism is
represented by
\begin{equation}
        \delta_\epsilon F
        =
        \{F,G[\epsilon]\}.
        \label{eq:gauge-generator-defining-property}
\end{equation}
The coefficients \(C_\alpha{}^\beta\) depend on the canonical variables
and ensure the correct transformations of the lapse and shift.  Their
explicit form will not be needed below.

Equation~\eqref{eq:Castellani-generator-schematic} fixes the structure of
the complete gauge generator, but the secondary quantities
\(\mathcal G_\alpha\) need to be found. We do this next.
We begin with the two tangential directions, for which the construction is
completely explicit.

\paragraph{Tangential diffeomorphisms.}

Set
\begin{equation}
        \epsilon^\perp=0 .
\end{equation}
The corresponding spacetime diffeomorphism is tangent to the spatial
slice.  On the dynamical canonical variables
\begin{equation}
        (\gamma_{ij},\rho^{ij};
        S^{ij},p_{ij}),
\end{equation}
the primary lapse--shift constraints in
\eqref{eq:Castellani-generator-schematic} act trivially.  They are
conjugate to the lapse and shift and therefore do not generate the
transformation of \(\gamma_{ij}\) or \(S^{ij}\).  Consequently the
tangential action of the full gauge generator on the dynamical phase
space is determined by its secondary part,
\begin{equation}
        G_\parallel[\epsilon]
        =
        \int_{\Sigma_t}\dd^2x\,
        \epsilon^i\mathcal G_i .
        \label{eq:G-tangential-secondary}
\end{equation}
We can determine \(\mathcal G_i\) as follows.  A spatial
diffeomorphism generated by an arbitrary spatial vector field \(\xi^i(x)\)
acts geometrically on the canonical configuration variables as
\begin{equation}
        \delta_\xi\gamma_{ij}
        =
        \mathcal L_\xi\gamma_{ij},
        \qquad
        \delta_\xi S^{ij}
        =
        \mathcal L_\xi S^{ij}.
        \label{eq:spatial-diffeomorphism-configurations}
\end{equation}
We therefore seek a phase-space functional \(H_{\rm sp}[\xi]\) satisfying
\begin{equation}
        \{\gamma_{ij},H_{\rm sp}[\xi]\}
        =
        \mathcal L_\xi\gamma_{ij},
        \qquad
        \{S^{ij},H_{\rm sp}[\xi]\}
        =
        \mathcal L_\xi S^{ij}.
        \label{eq:spatial-generator-defining-property}
\end{equation}
Because \(\rho^{ij}\) and \(p_{ij}\) are canonically conjugate to
\(\gamma_{ij}\) and \(S^{ij}\), respectively, these requirements are
implemented by
\begin{equation}
        H_{\rm sp}[\xi]
        :=
        \int_{\Sigma_t}\dd^2x
        \left(
        \rho^{ij}\mathcal L_\xi\gamma_{ij}
        +
        p_{ij}\mathcal L_\xi S^{ij}
        \right).
        \label{eq:spatial-generator-smeared}
\end{equation}
Indeed, the canonical Poisson brackets give 
\begin{equation}
        \{\gamma_{ij},H_{\rm sp}[\xi]\}
        =
        \mathcal L_\xi\gamma_{ij},
        \qquad
        \{S^{ij},H_{\rm sp}[\xi]\}
        =
        \mathcal L_\xi S^{ij}.
        \label{eq:spatial-generator-check}
\end{equation}
The corresponding momenta transform with their appropriate Lie
derivatives.  Thus \(H_{\rm sp}[\xi]\) is the canonical phase-space
generator of infinitesimal spatial diffeomorphisms.  To make this identification local, we now remove derivatives of the
smearing field \(\xi^i\).  For the spatial metric,
\begin{equation}
        \mathcal L_\xi\gamma_{ij}
        =
        D_i\xi_j+D_j\xi_i .
        \label{eq:Lie-gamma-spatial}
\end{equation}
Since \(S^{ij}\) is a contravariant spatial tensor density of weight
\(+1\),
\begin{equation}
\begin{aligned}
        \mathcal L_\xi S^{ij}
        ={}&
        \xi^kD_kS^{ij}
        -
        S^{kj}D_k\xi^i
        -
        S^{ik}D_k\xi^j
        +
        S^{ij}D_k\xi^k .
\end{aligned}
        \label{eq:Lie-S-spatial}
\end{equation}
Substituting
\eqref{eq:Lie-gamma-spatial} and
\eqref{eq:Lie-S-spatial} into
\eqref{eq:spatial-generator-smeared}, and integrating derivatives of
\(\xi^i\) by parts, gives
\begin{equation}
        H_{\rm sp}[\xi]
        =
        \int_{\Sigma_t}\dd^2x\,
        \xi^i\mathcal H_i ,
        \label{eq:Hsp-H_i}
\end{equation}
where \(\mathcal H_i\) is precisely the shift coefficient obtained from
the Legendre transformation in
Eq.~\eqref{eq:spatial-generator-explicit}.

Comparison of
\eqref{eq:G-tangential-secondary} with
\eqref{eq:Hsp-H_i} therefore identifies the two tangential secondary
generators:
\begin{equation}
        \mathcal G_i
        =
        \mathcal H_i,
        \qquad
        i=1,2 .
        \label{eq:G-i-equals-H-i}
\end{equation}
Thus the coefficient of the physical shift in the canonical Hamiltonian
is the generator density of spatial diffeomorphisms on the dynamical phase
space.

\paragraph{Normal diffeomorphisms.}

It remains to identify the normal secondary generator
\(\mathcal G_\perp\) appearing in
\eqref{eq:Castellani-generator-schematic}.  This direction is more subtle.

Before eliminating the auxiliary variables, the coefficient of the
physical lapse is
\(\mathcal H_N\) in
Eq.~\eqref{eq:HN-before-v-elimination}.  One might therefore expect simply
to identify
\(\mathcal G_\perp=\mathcal H_N\).  In the enlarged phase space, however,
the auxiliary variables carry second-class constraints, and the raw
coefficient \(\mathcal H_N\) need not have weakly vanishing Poisson
brackets with that second-class sector.  It therefore need not itself be
the appropriate first-class representative of the normal diffeomorphism.

Spacetime covariance has already established the normal gauge
direction.  Once the auxiliary second-class constraints are known, one may
modify the raw lapse coefficient by terms proportional to those
constraints so that the resulting functional has weakly vanishing brackets
with the second-class sector.  We denote such a representative by
\begin{equation}
        \mathcal H_\perp^\star .
\end{equation}
Accordingly, the normal secondary generator in the Castellani chain may be
chosen as
\begin{equation}
        \mathcal G_\perp
        =
        \mathcal H_\perp^\star .
        \label{eq:Gperp-Hperpstar}
\end{equation}
On the auxiliary second-class constraint surface this representative agrees
weakly with the raw lapse coefficient,
\begin{equation}
        \mathcal H_\perp^\star
        \weak
        \mathcal H_N .
        \label{eq:Hperpstar-HN}
\end{equation}
After the mixed auxiliary component \(v^i\) is eliminated below,
\(\mathcal H_N\) reduces to the density denoted by
\(\mathcal H_0\), and hence
\begin{equation}
        \mathcal H_\perp^\star
        \weak
        \mathcal H_0 .
        \label{eq:Hperpstar-H0}
\end{equation}
The projection required to construct
\(\mathcal H_\perp^\star\) is reviewed in
Appendix~\ref{app:improved-projection}.  The important point is that this
projection does not create a new gauge symmetry and does not alter the
number of gauge directions.  It merely chooses, within the equivalence
class that differs by auxiliary second-class constraints, a representative
appropriate for the reduced canonical description.

\paragraph{Structure of the complete diffeomorphism generator.}

We have therefore identified the secondary generators appearing in
\eqref{eq:Castellani-generator-schematic} as
\begin{equation}
        \mathcal G_\perp
        =
        \mathcal H_\perp^\star,
        \qquad
        \mathcal G_i
        =
        \mathcal H_i .
        \label{eq:secondary-diffeomorphism-generators-final}
\end{equation}
The full diffeomorphism generator may consequently be written
schematically as
\begin{equation}
\begin{aligned}
        G[\epsilon]
        =
        \int_{\Sigma_t}\dd^2x
        \Big[
        &\dot\epsilon^\alpha\Phi_\alpha
        +
        \epsilon^\perp\mathcal H_\perp^\star
        +
        \epsilon^i\mathcal H_i
        +
        \epsilon^\alpha
        C_\alpha{}^\beta\Phi_\beta
        \Big],
        \qquad
        \alpha=\perp,1,2 .
\end{aligned}
        \label{eq:full-diffeomorphism-generator-summary}
\end{equation}
Thus the three independent spacetime gauge functions give three
two-stage first-class chains.  The two tangential secondary generators are
the explicit momentum constraints \(\mathcal H_i\), while the normal
secondary generator is represented by
\(\mathcal H_\perp^\star\).  This is the canonical realization of
spacetime diffeomorphism invariance that will be used in the subsequent
constraint count.

\subsection{Eliminating the mixed auxiliary component}
\label{sec:auxiliary-shift-elimination}

We next eliminate the mixed auxiliary component \(v^i\).  It is important
not to confuse \(v^i\) with the physical shift \(N^i\).  The latter is a
gauge variable associated with spacetime diffeomorphisms, whereas \(v^i\)
is a component of the auxiliary tensor density \(P^{\mu\nu}\).  Moreover,
\(v^i\) carries no time derivative.  We will now show that, on the
nondegenerate branch specified below, its equation of motion determines
\(v^i\) algebraically in terms of the canonical variables rather than
producing a new constraint among them.

All dependence on \(v^i\) in the lapse coefficient
\(\mathcal H_N\) is contained in
\begin{equation}
        \mathcal H_v
        =
        2v^iB_i
        +
        \frac{s}{2}(S^{-1})_{ij}v^iv^j,
        \label{eq:Hv-B-form}
\end{equation}
of which variation with respect to \(v^i\) gives
\begin{equation}
        2B_i
        +
        s(S^{-1})_{ij}v^j
        =
        0 ,
        \label{eq:v-equation-corrected}
\end{equation}
with the unique solution
\begin{equation}
        v_*^i
        =
        -\frac{2}{s}S^{ij}B_j
        =
        -\frac{2}{s}S^{ij}D_kp^k{}_j .
        \label{eq:v-solution-corrected}
\end{equation}
Thus \(v^i\) is completely determined by the canonical variables on
\(s\neq0\).  Substitution of \(v^i=v_*^i\) simply  gives
\begin{equation}
        \mathcal H_v\big|_{v=v_*}
        =
        -\frac{2}{s}S^{ij}B_iB_j .
        \label{eq:Hv-solved}
\end{equation}
Equivalently,
\begin{equation}
        \mathcal H_v^{\rm red}
        =
        -\frac{2}{s}
        S^{ij}
        \left(D_kp^k{}_i\right)
        \left(D_lp^l{}_j\right).
        \label{eq:Hv-reduced-explicit}
\end{equation}
The important point is that the equation for \(v^i\) does not impose an
additional constraint on
\[
        (\gamma_{ij},\rho^{ij};S^{ij},p_{ij}) .
\]
Instead, it determines the auxiliary variable \(v^i\) algebraically.
Eliminating \(v^i\) therefore replaces the entire \(v^i\)-dependent part
of \(\mathcal H_N\) by the reduced contribution
\eqref{eq:Hv-reduced-explicit}.

For completeness, the same conclusion can be expressed in unreduced
Dirac language.  If \(v^i\) and its conjugate momentum
\(\pi_i^{(v)}\) are retained as canonical variables, the absence of
\(\dot v^i\) produces the two primary constraints
\begin{equation}
        \pi_i^{(v)}
        \weak0 .
\end{equation}
Their preservation gives the two algebraic conditions
\begin{equation}
        2B_i
        +
        s(S^{-1})_{ij}v^j
        \weak0 .
\end{equation}
Their mutual Poisson-bracket block contains the invertible matrix
\begin{equation}
        s(S^{-1})_{ij},
\end{equation}
and is therefore nondegenerate when \(s\neq0\).  The two primary
constraints and the two algebraic conditions consequently form a
four-dimensional second-class block. Solving Eq.~\eqref{eq:v-equation-corrected} and substituting
\(v^i=v_*^i\) is precisely the corresponding second-class reduction.
As shown in Sec.~\ref{sec:reduced-extended-bookkeeping}, this elimination
does not modify the canonical brackets
\eqref{eq:gamma-rho-bracket}--\eqref{eq:S-p-bracket} between functionals
that are independent of \(v^i\).

\subsection{Reduced Hamiltonian and auxiliary normal constraint}
\label{sec:reduced-Hamiltonian-normal-constraint}

Having eliminated the algebraic mixed auxiliary component \(v^i\), we now
arrive at the canonical system on which the remaining Dirac analysis will
be performed.  Up to the overall normalization of the action,
\begin{equation}
 I_{\rm red}
 =
 \int\dd t\,\dd^2x
 \left[
 \rho^{ij}\dot\gamma_{ij}
 +
 p_{ij}\dot S^{ij}
 -
 N\mathcal H_0
 -
 N^i\mathcal H_i
 -
 \nu\mathcal C_0
 \right].
 \label{eq:full-reduced-canonical-action}
\end{equation}
The dynamical canonical pairs are still
\begin{equation}
        (\gamma_{ij},\rho^{ij}),
        \qquad
        (S^{ij},p_{ij}),
\end{equation}
whereas \(N\), \(N^i\), and \(\nu\) occur without time derivatives.
The lapse and shift implement spacetime diffeomorphisms, while the role of
\(\nu\) is different: it imposes the auxiliary scalar constraint
\(\mathcal C_0\weak0\). The reduced lapse coefficient is obtained from the unreduced one by
substituting the algebraic solution \(v^i=v_*^i\):
\begin{equation}
        \mathcal H_0
        :=
        \left.
        \mathcal H_N
        \right|_{v=v_*}.
        \label{eq:H0-from-HN}
\end{equation}
Using the reduced \(v^i\)-contribution
\eqref{eq:Hv-solved}, the lapse coefficient becomes
\begin{equation}
\begin{aligned}
 \mathcal H_0
 ={}&
 \beta\sqrt\gamma
 -
 \frac12S
 -
 \frac2c\pi^{ij}\Pi_{ij}
 -
 \frac1c\mathcal Q(S,\Pi)
 \\
 &-
 \frac2sS^{ij}B_iB_j
 +
 cD_iD_j\mathcal A^{ij}.
\end{aligned}
\label{eq:H0-full-tensor}
\end{equation}

Notice also the density weights.  The quantity \(\mathcal C_0\) is a
spatial scalar, while \(\nu\) is a scalar density of weight \(+1\).
Consequently \(\nu\mathcal C_0\) is a spatial density, as required in the
canonical action.

\paragraph{The derivative term in \(\mathcal H_0\).}

The last term in Eq.~\eqref{eq:H0-full-tensor} is particularly important: For a constant lapse its integral is a boundary term.  For an arbitrary
lapse \(N(x)\), however, it cannot be discarded: it contributes to the
functional derivatives of the smeared lapse generator
\begin{equation}
        H_0[N]
        :=
        \int_{\Sigma_t}\dd^2x\,N\mathcal H_0 .
        \label{eq:H0-smeared-definition}
\end{equation}
This term will be essential when we impose preservation of
\(\mathcal C_0\). 

The reduced lapse Hamiltonian in the canonical variables is 
\begin{align}
 H_0[N]
 =\int_{\Sigma_t}\dd^2x\,N
 \Bigg\{&
 \beta\sqrt\gamma
 -\frac12S
 \nonumber\\
 &-\frac2c
 \left[
 \rho^{ij}p_{ij}
 -p\rho
 -pS^{ij}p_{ij}
 +S\left(p^2-p_{ij}p^{ij}\right)
 \right]
 \nonumber\\
 &-\frac1c\,\mathcal Q_p(S,p)
 -\frac2sS^{ij}B_iB_j
 \nonumber\\
 &+c
 \left(
 D_iD_jS^{ij}-D^2S
 \right)
 \Bigg\},
 \label{eq:H0-final-canonical-variables}
\end{align}
where the remaining quadratic algebraic function is
\begin{equation}
 \mathcal Q_p(S,p)
 :=
 \frac{3S}{2}
 \left(
 p_{ij}p^{ij}-p^2
 \right)
 +
 3pS^{ij}p_{ij}
 -
 2S^{ij}p_i{}^kp_{kj}.
 \label{eq:Qp-final}
\end{equation}
Thus every quantity appearing in
Eq.~\eqref{eq:H0-final-canonical-variables} is a function of
\((\gamma_{ij},\rho^{ij};S^{ij},p_{ij})\).

\paragraph{Vacuum check.}

As a simple consistency check, consider a time-symmetric slice of the
maximally symmetric BINMG vacuum.  There
\begin{equation}
        p_{ij}=0,
        \qquad
        s=a,
        \qquad
        K_{ij}=0 .
\end{equation}
The normal--normal Einstein projection then gives
\begin{equation}
        G_{\perp\perp}
        =
        \frac12\,{}^{(2)}R
        =
        \Lambda .
\end{equation}
Therefore
\begin{equation}
        \bar{\mathcal C}_0
        =
        \frac12(1-a)-c\Lambda .
        \label{eq:scalar-constraint-vac}
\end{equation}
Using the BINMG vacuum relation  \eqref{eq:two-c-Lambda},
we obtain
\begin{equation}
        \bar{\mathcal C}_0=0 .
\end{equation}
Thus the maximally symmetric vacuum lies on the nonlinear scalar
constraint surface, as required.

\subsection{Preservation of the scalar constraint}
\label{sec:scalar-preservation}

The scalar condition
\begin{equation}
        \mathcal C_0\weak0
\end{equation}
by itself does not remove an entire canonical pair.  The next step in the
Dirac algorithm is therefore to impose its preservation,
\begin{equation}
        \dot{\mathcal C}_0\weak0,
\end{equation}
and determine what this condition implies.

There are three logical possibilities.  Preservation of
\(\mathcal C_0\) could determine the auxiliary multiplier \(\nu\); it
could impose an equation on the physical lapse \(N\); or it could generate
an additional scalar constraint.  We will show that the third possibility
is realized:
\begin{equation}
        \{\mathcal C_0(x),\mathcal C_0(y)\}=0
\end{equation}
strongly, so \(\nu\) is not fixed at this stage, while
\begin{equation}
        \{\mathcal C_0(x),H_0[N]\}
        =
        N(x)\mathcal C_1(x)
\end{equation}
contains no spatial derivative of \(N\).  Hence preservation of
\(\mathcal C_0\) produces the secondary scalar constraint
\(\mathcal C_1\weak0\), without determining either \(\nu\) or the
physical lapse.

The reduced Hamiltonian is
\begin{equation}
        H_{\rm red}
        =
        H_0[N]
        +
        H_{\rm sp}[N^i]
        +
        \mathcal C_0[\nu],
        \label{eq:Hred-smeared-scalar-analysis}
\end{equation}
where
\begin{equation}
        H_0[N]
        :=
        \int_{\Sigma_t}\dd^2x\,N\mathcal H_0 ,
\end{equation}
and
\begin{equation}
        \mathcal C_0[\nu]
        :=
        \int_{\Sigma_t}\dd^2x\,\nu\mathcal C_0 .
\end{equation}

\paragraph{ The self-bracket of \(\mathcal C_0\).}

We first determine whether preservation of \(\mathcal C_0\) fixes its
multiplier \(\nu\).  Introduce compactly supported smearing densities
\(f\) and \(g\), both of spatial weight \(+1\), and define
\begin{equation}
        \mathcal C_0[f]
        :=
        \int_{\Sigma_t}\dd^2x\,f\mathcal C_0 .
        \label{eq:C0-smeared}
\end{equation}
It is useful to split the scalar constraint \eqref{eq:C0-before-v-elimination} into
\begin{equation}
        \mathcal C_0
        =
        \mathcal V+W+\mathcal R,
        \label{eq:C0-UWR}
\end{equation}
with
\begin{equation}
        \mathcal V
        =
        \frac12(1-s),
        \qquad
        W
        =
        -\frac1{2c}
        \left(
        p^2-p_{ij}p^{ij}
        \right),
        \qquad
        \mathcal R
        =
        -\frac c2\,{}^{(2)}R .
        \label{eq:UWR-definitions}
\end{equation}
The dependence of these three terms on the canonical variables is
particularly simple:
\[
        \mathcal V=\mathcal V[S],
        \qquad
        W=W[\gamma,p],
        \qquad
        \mathcal R=\mathcal R[\gamma].
\]
None depends on \(\rho^{ij}\).  Hence the canonical pair
\((\gamma_{ij},\rho^{ij})\) gives no contribution to the self-bracket.
Within the pair \((S^{ij},p_{ij})\), the only possible nonzero contribution
is the cross-bracket between \(\mathcal V\) and \(W\).

Using
\begin{equation}
        \delta s
        =
        s(S^{-1})_{ij}\delta S^{ij},
        \label{eq:delta-s-identity}
\end{equation}
we obtain
\begin{equation}
        \frac{\delta\mathcal V[f]}
             {\delta S^{ij}}
        =
        -\frac12fs(S^{-1})_{ij},
        \label{eq:U-derivative}
\end{equation}
while
\begin{equation}
        \frac{\delta W[g]}
             {\delta p_{ij}}
        =
        \frac gc
        \left(
        p^{ij}-p\gamma^{ij}
        \right).
        \label{eq:W-derivative}
\end{equation}
Therefore
\begin{equation}
        \{\mathcal V[f],W[g]\}
        =
        -\int_{\Sigma_t}\dd^2x\,
        \frac{fg\,s}{2c}
        (S^{-1})_{ij}
        \left(
        p^{ij}-p\gamma^{ij}
        \right).
        \label{eq:VW-bracket}
\end{equation}
The right-hand side is symmetric under \(f\leftrightarrow g\).  Since
\begin{equation}
        \{W[f],\mathcal V[g]\}
        =
        -\{\mathcal V[g],W[f]\},
\end{equation}
the two cross-terms cancel.  Hence
\begin{equation}
        \{\mathcal C_0[f],\mathcal C_0[g]\}=0
        \label{eq:C0-self-smeared-zero}
\end{equation}
for arbitrary compactly supported smearings, or equivalently
\begin{equation}
        \{\mathcal C_0(x),\mathcal C_0(y)\}=0 .
        \label{eq:C0-self-local-zero}
\end{equation}

This is a strong identity.  Consequently the term
\(\mathcal C_0[\nu]\) in the Hamiltonian makes no contribution to
\(\dot{\mathcal C}_0\), and the multiplier \(\nu\) is not determined by
preservation of \(\mathcal C_0\).

\paragraph{The bracket with  lapse.}

We next study
\begin{equation}
        \{\mathcal C_0(x),H_0[N]\}.
\end{equation}
The issue is not whether the bracket is nonzero; generically it is.
The important question is whether it contains spatial derivatives of the
lapse.  If terms proportional to \(D_iN\) or \(D_iD_jN\) survived, the
condition \(\dot{\mathcal C}_0\weak0\) could become a differential equation
for \(N\).  We now show explicitly that all such terms cancel.

\begin{proposition}[Arbitrary-lapse factorization]
\label{prop:arbitrary-lapse-factorization}
On the regular region \(s\neq0\), modulo the spatial boundary terms
excluded in the local analysis,
\begin{equation}
        \{\mathcal C_0(x),H_0[N]\}
        =
        N(x)\mathcal C_1(x),
        \label{eq:C1-factorization}
\end{equation}
where
\begin{equation}
        \mathcal C_1(x)
        :=
        \{\mathcal C_0(x),H_0[1]\}.
        \label{eq:C1-definition}
\end{equation}
\end{proposition}

\begin{proof}

To identify all possible spatial derivatives of the lapse, decompose \eqref{eq:H0-final-canonical-variables}
\begin{equation}
\begin{aligned}
        H_0[N]
        ={}&
        H_{\rm pot}[N]
        +
        H_S[N]
        +
        H_{\rm mix}[N]
        +
        H_{\mathcal Q}[N]
        +
        H_B[N]
        +
        H_D[N].
        \label{eq:H0-six-pieces}
\end{aligned}
\end{equation}
The individual pieces are
\begin{align}
        H_{\rm pot}[N]
        &:=
        \int_{\Sigma_t}\dd^2x\,
        N\beta\sqrt\gamma ,
        \\
        H_S[N]
        &:=
        -\frac12
        \int_{\Sigma_t}\dd^2x\,
        NS ,
        \\
        H_{\rm mix}[N]
        &:=
        -\frac2c
        \int_{\Sigma_t}\dd^2x\,N
        \left[
        \rho^{ij}p_{ij}
        -p\rho
        -pS^{ij}p_{ij}
        +S\left(p^2-p_{ij}p^{ij}\right)
        \right],
        \label{eq:Hmix-final-proof}
        \\
        H_{\mathcal Q}[N]
        &:=
        -\frac1c
        \int_{\Sigma_t}\dd^2x\,
        N\mathcal Q_p(S,p),
        \\
        H_B[N]
        &:=
        -2
        \int_{\Sigma_t}\dd^2x\,
        N\frac{S^{ij}}sB_iB_j,
        \label{eq:HB-definition-lapse-proof}
        \\
        H_D[N]
        &:=
        c
        \int_{\Sigma_t}\dd^2x\,
        N
        \left(
        D_iD_jS^{ij}-D^2S
        \right).
        \label{eq:HD-definition-lapse-proof}
\end{align}

We now ask a very specific question: which terms in
\(\{\mathcal C_0,H_0[N]\}\) can produce derivatives of the arbitrary
lapse \(N(x)\)?

There are only three possibilities.

First, the curvature term
\begin{equation}
        \mathcal R
        =
        -\frac c2\,{}^{(2)}R
\end{equation}
contains two spatial derivatives of \(\gamma_{ij}\).
Its bracket with the \(\rho^{ij}\)-dependent part of
\(H_{\rm mix}[N]\) can therefore transfer derivatives onto \(N\).

Second,
\begin{equation}
        B_i=D_jp^j{}_i
\end{equation}
contains one spatial derivative of the canonical momentum.
The \(p_{ij}\)-variation of \(H_B[N]\) can therefore generate
\(D_iN\).  This contribution pairs with the algebraic
\(S^{ij}\)-dependence of \(\mathcal V\).

Third, integrating the last term of the Hamiltonian twice by parts gives
\begin{equation}
\begin{aligned}
        H_D[N]
        &=
        c
        \int_{\Sigma_t}\dd^2x\,
        N
        \left(
        D_iD_jS^{ij}-D^2S
        \right)
        \nonumber\\
        &\doteq
        c
        \int_{\Sigma_t}\dd^2x\,
        \left(
        S^{ij}-S\gamma^{ij}
        \right)
        D_iD_jN .
        \label{eq:HD-integrated-proof}
\end{aligned}
\end{equation}
Its variation with respect to \(S^{ij}\) therefore contains
\(D_iD_jN\), which can pair with the \(p_{ij}\)-dependence of \(W\).

There are no other sources of derivatives of the lapse.
The terms \(H_{\rm pot}\), \(H_S\), and \(H_{\mathcal Q}\) are
algebraic in the canonical variables, while the remaining variations
of \(H_{\rm mix}\) and the \(S^{ij}\)-variation of \(H_B\) contain
\(N\) without spatial derivatives.

We now evaluate the three possible derivative-of-lapse contributions.

\medskip
\noindent
\textit{(i) Curvature--mixed contribution.}

From Eq.~\eqref{eq:Hmix-final-proof},
\begin{equation}
        \frac{\delta H_{\rm mix}[N]}
             {\delta\rho^{ij}}
        =
        -\frac2c\,N
        \left(
        p_{ij}-p\gamma_{ij}
        \right).
        \label{eq:Hmix-rho-final}
\end{equation}
The variation of the two-dimensional Ricci scalar with respect to the
covariant spatial metric is
\begin{equation}
        \delta{}^{(2)}R
        =
        \left(
        D^iD^j
        -
        \gamma^{ij}D^2
        -
        {}^{(2)}R^{ij}
        \right)
        \delta\gamma_{ij}.
        \label{eq:R2-variation-covariant}
\end{equation}
Keeping only terms in which derivatives act on \(N\), one obtains
\begin{align}
 \left.
 \{\mathcal R,H_{\rm mix}[N]\}
 \right|_{\partial N}
 ={}&
 2D_iN
 \left[
 D_j
 \left(
 p^{ij}-p\gamma^{ij}
 \right)
 +D^ip
 \right]
 \nonumber\\
 &+
 \left[
 \left(
 p^{ij}-p\gamma^{ij}
 \right)
 +p\gamma^{ij}
 \right]
 D_iD_jN .
 \label{eq:curvature-mixed-final-intermediate}
\end{align}
Metric compatibility gives
\begin{equation}
        D_j
        \left(
        p^{ij}-p\gamma^{ij}
        \right)
        +D^ip
        =
        D_jp^{ij}
        =
        B^i,
\end{equation}
while the second bracket in
Eq.~\eqref{eq:curvature-mixed-final-intermediate} is simply
\(p^{ij}\).  Hence
\begin{equation}
 \left.
 \{\mathcal R,H_{\rm mix}[N]\}
 \right|_{\partial N}
 =
 2B^iD_iN
 +
 p^{ij}D_iD_jN .
 \label{eq:lapse-derivatives-curvature-result}
\end{equation}

\medskip
\noindent
\textit{(ii) Algebraic--\(B_i\) contribution.}

For
\begin{equation}
        \mathcal V
        =
        \frac12(1-s),
\end{equation}
we have
\begin{equation}
        \frac{\delta\mathcal V}{\delta S^{ab}}
        =
        -\frac12s(S^{-1})_{ab}.
        \label{eq:V-S-derivative-final}
\end{equation}
Variation of \(H_B[N]\) with respect to \(p_{ab}\) gives
\begin{equation}
        \frac{\delta H_B[N]}{\delta p_{ab}}
        =
        2D^a
        \left(
        N\frac{S^{ib}}sB_i
        \right)
        +
        2D^b
        \left(
        N\frac{S^{ia}}sB_i
        \right).
        \label{eq:HB-p-variation}
\end{equation}
Therefore
\begin{equation}
        \{\mathcal V,H_B[N]\}
        =
        -2s(S^{-1})_{ab}
        D^a
        \left(
        N\frac{S^{ib}}sB_i
        \right).
        \label{eq:VHB-bracket}
\end{equation}
The part containing a derivative of the lapse is
\begin{equation}
 \left.
 \{\mathcal V,H_B[N]\}
 \right|_{\partial N}
 =
 -2B^iD_iN .
 \label{eq:lapse-derivatives-shift-result}
\end{equation}

\medskip
\noindent
\textit{(iii) Momentum--derivative contribution.}

For
\begin{equation}
        W
        =
        -\frac1{2c}
        \left(
        p^2-p_{ij}p^{ij}
        \right),
\end{equation}
the momentum derivative is
\begin{equation}
        \frac{\delta W}{\delta p_{ab}}
        =
        \frac1c
        \left(
        p^{ab}-p\gamma^{ab}
        \right).
        \label{eq:W-p-derivative-final}
\end{equation}
From Eq.~\eqref{eq:HD-integrated-proof},
\begin{equation}
        \frac{\delta H_D[N]}{\delta S^{ab}}
        =
        c
        \left(
        D_aD_bN
        -
        \gamma_{ab}D^2N
        \right).
        \label{eq:HD-S-derivative-final}
\end{equation}
Using the \(S^{ij},p_{ij}\) part of the canonical Poisson bracket,
\begin{equation}
        \{F,G\}_{S,p}
        =
        \int\dd^2x
        \left(
        \frac{\delta F}{\delta S^{ij}}
        \frac{\delta G}{\delta p_{ij}}
        -
        \frac{\delta F}{\delta p_{ij}}
        \frac{\delta G}{\delta S^{ij}}
        \right),
\end{equation}
we obtain
\begin{align}
        \{W,H_D[N]\}
        &=
        -
        \left(
        p^{ab}-p\gamma^{ab}
        \right)
        \left(
        D_aD_bN-\gamma_{ab}D^2N
        \right)
        \nonumber\\
        &=
        -p^{ab}D_aD_bN .
        \label{eq:lapse-derivatives-acceleration-nonlinear}
\end{align}
In the last step we used the fact that the spatial dimension is two.

Thus
\begin{equation}
        \{W,H_D[N]\}
        =
        -p^{ab}D_aD_bN .
\end{equation}

The cancellation of every derivative of the lapse is now manifest:
\begin{align}
 &\left.
 \{\mathcal R,H_{\rm mix}[N]\}
 \right|_{\partial N}
 +
 \left.
 \{\mathcal V,H_B[N]\}
 \right|_{\partial N}
 +
 \{W,H_D[N]\}
 \nonumber\\[1mm]
 &\qquad=
 \left(
 2B^iD_iN+p^{ij}D_iD_jN
 \right)
 -
 2B^iD_iN
 -
 p^{ij}D_iD_jN
 \nonumber\\
 &\qquad=0 .
 \label{eq:all-lapse-derivatives-cancel}
\end{align}
Therefore
\begin{equation}
        \left.
        \{\mathcal C_0,H_0[N]\}
        \right|_{\partial N}
        =
        0 .
        \label{eq:no-lapse-derivatives-final}
\end{equation}

The bracket is linear in \(N\), and no spatial derivative of \(N\)
survives.  Since the theory is local, the bracket must therefore take
the form
\begin{equation}
        \{\mathcal C_0(x),H_0[N]\}
        =
        N(x)F(x)
        \label{eq:C0H0-local-F}
\end{equation}
for some local phase-space function
\begin{equation}
        F
        =
        F[\gamma_{ij},\rho^{ij},S^{ij},p_{ij}].
\end{equation}
Setting \(N=1\) in the identity gives
\begin{equation}
        F(x)
        =
        \{\mathcal C_0(x),H_0[1]\}
        =
        \mathcal C_1(x),
\end{equation}
and therefore
\begin{equation}
        \{\mathcal C_0(x),H_0[N]\}
        =
        N(x)\mathcal C_1(x).
\end{equation}
This proves Eq.~\eqref{eq:C1-factorization}.
\end{proof}

\paragraph{ The complete consistency condition.}

We now restore the remaining terms in the reduced Hamiltonian,
\begin{equation}
        H_{\rm red}
        =
        H_0[N]
        +
        H_{\rm sp}[N^i]
        +
        \mathcal C_0[\nu].
\end{equation}
The spatial-diffeomorphism generator is written directly in the final
canonical variables as
\begin{equation}
        H_{\rm sp}[N^i]
        =
        \int_{\Sigma_t}\dd^2x
        \left[
        \rho^{ij}\mathcal L_{\vec N}\gamma_{ij}
        +
        p_{ij}\mathcal L_{\vec N}S^{ij}
        \right].
        \label{eq:Hsp-final-canonical}
\end{equation}
Since \(\mathcal C_0\) is a spatial scalar, this gives
\begin{equation}
        \{\mathcal C_0,H_{\rm sp}[N^i]\}
        =
        \mathcal L_{\vec N}\mathcal C_0
        =
        N^iD_i\mathcal C_0 .
        \label{eq:C0-spatial-diffeomorphism}
\end{equation}
This term therefore vanishes weakly on
\(\mathcal C_0\weak0\).
 Consequently,
\begin{equation}
        0
        \weak
        \dot{\mathcal C}_0
        =
        N\mathcal C_1 .
        \label{eq:C0-complete-consistency}
\end{equation}
On a regular ADM chart,
\begin{equation}
        N\neq0,
\end{equation}
and hence
\begin{equation}
        \mathcal C_1\weak0 .
        \label{eq:C1-secondary}
\end{equation}
The condition \(N\neq0\) is only the regularity condition required for
the ADM foliation. 

We have therefore reached the second member of the scalar Dirac chain,
\begin{equation}
        \mathcal C_0
        \ \longrightarrow\
        \mathcal C_1 .
\end{equation}
Preservation of \(\mathcal C_0\) fixes neither the auxiliary field
\(\nu\) nor the physical lapse \(N\).  Instead, it generates a new
scalar constraint.

The next question is whether \(\mathcal C_1\) is independent and whether
the pair
\((\mathcal C_0,\mathcal C_1)\) is second class.  This is determined by
their mutual Poisson bracket.

\subsection{The scalar Dirac matrix}
\label{sec:scalar-dirac}

The Dirac algorithm has produced two scalar constraints,
\begin{equation}
        \mathcal C_0\weak0,
        \qquad
        \mathcal C_1\weak0 .
\end{equation}
To determine whether they remove one canonical pair, we must determine the
rank of their mutual Poisson-bracket matrix.  If the scalar Dirac matrix is
invertible, then \((\mathcal C_0,\mathcal C_1)\) is a second-class pair.
Preservation of \(\mathcal C_1\) then fixes the auxiliary multiplier
\(\nu\), and the scalar Dirac chain terminates.

The scalar Dirac matrix is
\begin{equation}
        \mathbb D_{AB}(x,y)
        =
        \begin{pmatrix}
        \{\mathcal C_0(x),\mathcal C_0(y)\}
        &
        \{\mathcal C_0(x),\mathcal C_1(y)\}
        \\[1mm]
        \{\mathcal C_1(x),\mathcal C_0(y)\}
        &
        \{\mathcal C_1(x),\mathcal C_1(y)\}
        \end{pmatrix}.
        \label{eq:scalar-Dirac-matrix}
\end{equation}
The first entry vanishes strongly,
\begin{equation}
        \{\mathcal C_0(x),\mathcal C_0(y)\}=0 ,
\end{equation}
so the essential quantity is the off-diagonal kernel
\begin{equation}
        \Delta(x,y)
        :=
        \{\mathcal C_0(x),\mathcal C_1(y)\}.
        \label{eq:Delta-definition}
\end{equation}
If \(\Delta\) is invertible, then the full matrix
\(\mathbb D\) is invertible independently of the value of
\(\{\mathcal C_1,\mathcal C_1\}\).

At this point there is a potentially important complication.
Since \(\mathcal C_0\) and \(\mathcal C_1\) contain spatial derivatives of
the canonical variables, their bracket could in principle have the form
\begin{equation}
        \Delta(x,y)
        =
        d_0(x)\delta^{(2)}(x-y)
        +
        d_1^i(x)D_i\delta^{(2)}(x-y)
        +
        \cdots .
\end{equation}
The rank question would then be a differential-operator problem rather
than a pointwise algebraic one.  The arbitrary-lapse identity established
in Proposition~\ref{prop:arbitrary-lapse-factorization} rules this out
completely.

\begin{proposition}[Ultralocality of the scalar Dirac kernel]
\label{prop:Delta-exact-ultralocal}
Assume \(s\neq0\) and the exact arbitrary-lapse identity
\begin{equation}
        \{\mathcal C_0(x),H_0[N]\}
        =
        N(x)\mathcal C_1(x)
        \label{eq:C0-H0-factorization-recall}
\end{equation}
for every smooth lapse \(N\).  Then
\begin{equation}
        \Delta(x,y)
        =
        \mathfrak d(x)\delta^{(2)}(x-y)
        \label{eq:Delta-exact-multiplication}
\end{equation}
for some local phase-space function \(\mathfrak d(x)\).  In particular,
no derivative of the delta function occurs.
\end{proposition}

\begin{proof}
The proof uses only the two strong identities
\begin{equation}
        \{\mathcal C_0(x),\mathcal C_0(y)\}=0
        \label{eq:C0-self-recall-Delta}
\end{equation}
and
\begin{equation}
        \{\mathcal C_0(x),H_0[N]\}
        =
        N(x)\mathcal C_1(x).
        \label{eq:C0-H0-recall-Delta}
\end{equation}

Apply the Jacobi identity to
\(\mathcal C_0(x)\), \(\mathcal C_0(y)\), and \(H_0[N]\).
Because the self-bracket
\eqref{eq:C0-self-recall-Delta} vanishes strongly, we obtain
\begin{equation}
        \{\mathcal C_0(x),
        \{\mathcal C_0(y),H_0[N]\}\}
        =
        \{\mathcal C_0(y),
        \{\mathcal C_0(x),H_0[N]\}\}.
        \label{eq:Jacobi-Delta-step}
\end{equation}
Using \eqref{eq:C0-H0-recall-Delta} on both sides gives
\begin{equation}
        N(y)\Delta(x,y)
        =
        N(x)\Delta(y,x).
        \label{eq:N-Delta-relation-before-symmetry}
\end{equation}

Setting \(N=1\) shows first that
\begin{equation}
        \Delta(x,y)
        =
        \Delta(y,x).
        \label{eq:Delta-symmetric-kernel}
\end{equation}
Returning to an arbitrary lapse therefore gives
\begin{equation}
        \bigl[N(y)-N(x)\bigr]\Delta(x,y)=0
        \label{eq:N-difference-Delta-zero}
\end{equation}
for every smooth function \(N\).

Now regard \(\Delta\) as an operator acting on compactly supported test
functions:
\begin{equation}
        (\Delta f)(x)
        :=
        \int_{\Sigma_t}\dd^2y\,
        \Delta(x,y)f(y).
        \label{eq:Delta-operator-definition}
\end{equation}
Equation~\eqref{eq:N-difference-Delta-zero} is equivalent to
\begin{equation}
        \Delta(Nf)
        =
        N\,\Delta f
        \label{eq:Delta-commutes-multiplication}
\end{equation}
for every smooth \(N\).  Thus \(\Delta\) commutes with multiplication by
every smooth function.

By locality, \(\Delta\) is a finite-order differential operator on the
test function.  In a local coordinate chart it may therefore be written
as
\begin{equation}
        \Delta
        =
        \sum_{|\alpha|\le r}
        b_\alpha(x)\partial^\alpha .
        \label{eq:Delta-chart-expansion}
\end{equation}
If \(r\ge1\), the commutator
\begin{equation}
        [\Delta,N]f
        =
        \Delta(Nf)-N\Delta f
\end{equation}
contains derivatives of \(N\).  Since
Eq.~\eqref{eq:Delta-commutes-multiplication} requires this commutator to
vanish for every smooth \(N\), the coefficients multiplying all
positive-order derivatives must vanish.  Hence only the zeroth-order term
remains:
\begin{equation}
        \Delta f
        =
        \mathfrak d(x)f(x).
\end{equation}
Equivalently,
\begin{equation}
        \Delta(x,y)
        =
        \mathfrak d(x)\delta^{(2)}(x-y),
\end{equation}
which proves the proposition.
\end{proof}
The scalar-rank problem has therefore become pointwise and algebraic.
Wherever
\begin{equation}
        \mathfrak d(x)\neq 0
        \label{eq:d-nonzero-condition}
\end{equation}
the kernel is invertible, with
\begin{equation}
        \Delta^{-1}(x,y)
        =
        \frac{1}{\mathfrak d(x)}
        \delta^{(2)}(x-y).
        \label{eq:Delta-local-inverse}
\end{equation}

\paragraph{The scalar pair is second class.}

It is not necessary to compute
\(\{\mathcal C_1,\mathcal C_1\}\).
Define
\begin{equation}
        \Gamma(x,y)
        :=
        \{\mathcal C_1(x),\mathcal C_1(y)\}.
\end{equation}
Then
\begin{equation}
        \mathbb D
        =
        \begin{pmatrix}
        0 & \Delta
        \\
        -\Delta^{T} & \Gamma
        \end{pmatrix}.
        \label{eq:scalar-Dirac-block-form}
\end{equation}
Suppose that \((u,v)\) lies in the kernel of \(\mathbb D\).  The first
row gives
\begin{equation}
        \Delta v=0 .
\end{equation}
If \(\mathfrak d\neq0\), then \(\Delta\) is invertible, and hence
\begin{equation}
        v=0 .
\end{equation}
The second row then reduces to
\begin{equation}
        \Delta^T u=0 ,
\end{equation}
which similarly implies \(u=0\).  Thus the full scalar Dirac matrix has no
kernel.  Consequently,
\begin{equation}
        \mathcal C_0,\mathcal C_1
        \quad\hbox{form a second-class pair whenever}\quad
        \mathfrak d\neq0 .
        \label{eq:C0-C1-second-class-condition}
\end{equation}

\paragraph{Preservation of \(\mathcal C_1\).}

It remains to impose the final consistency condition.  Using the reduced
Hamiltonian,
\begin{equation}
\begin{aligned}
        0
        \weak
        \dot{\mathcal C}_1(x)
        =
        \{\mathcal C_1(x),H_0[N]\}
        +
        \{\mathcal C_1(x),H_{\rm sp}[N^i]\}+
        \int_{\Sigma_t}\dd^2y\,
        \nu(y)
        \{\mathcal C_1(x),\mathcal C_0(y)\}.
\end{aligned}
        \label{eq:C1-full-preservation}
\end{equation}

Since \(\mathcal C_1\) is a spatial scalar, the second term is its Lie
derivative along the shift and therefore vanishes weakly on
\(\mathcal C_1\weak0\).  Define
\begin{equation}
        F_N(x)
        :=
        \{\mathcal C_1(x),H_0[N]\}.
        \label{eq:FN-definition}
\end{equation}
Using
\begin{equation}
        \{\mathcal C_1(x),\mathcal C_0(y)\}
        =
        -\Delta(y,x)
        =
        -\mathfrak d(x)\delta^{(2)}(x-y),
\end{equation}
where the symmetry
\eqref{eq:Delta-symmetric-kernel} has been used, the preservation equation
becomes
\begin{equation}
        0
        \weak
        F_N(x)-\mathfrak d(x)\nu(x).
        \label{eq:C1-preservation-d}
\end{equation}
Hence, wherever \(\mathfrak d\neq0\),
\begin{equation}
        \nu(x)
        =
        \frac{F_N(x)}{\mathfrak d(x)} .
        \label{eq:nu-fixed-by-C1}
\end{equation}

This equation determines the auxiliary variable \(\nu\), not the physical
lapse \(N\).  The quantity \(F_N\) may depend on \(N\) and on spatial
derivatives of \(N\), but for every freely chosen lapse the equation above
returns the corresponding value of \(\nu\).  Thus the physical lapse
remains arbitrary, as required by diffeomorphism invariance.

We have therefore established the complete scalar Dirac mechanism:
\begin{equation}
        \mathcal C_0
        \ \longrightarrow\
        \mathcal C_1,
        \qquad
        \{\mathcal C_0(x),\mathcal C_1(y)\}
        =
        \mathfrak d(x)\delta^{(2)}(x-y),
        \qquad
        \mathfrak d\neq0
\end{equation}
implies that the two scalar constraints are second class and that
preservation of \(\mathcal C_1\) fixes \(\nu\).  The scalar chain therefore
terminates without fixing the physical lapse.

It remains only to establish that
\(\mathfrak d\) is nonzero on a physically relevant region of phase space.
We do this by evaluating it at the maximally symmetric BINMG vacuum.

\subsection{Rank of the scalar bracket at the maximally symmetric vacuum}
\label{sec:vacuum-scalar-rank}

The preceding subsection reduced the scalar-rank problem to the local
coefficient
\begin{equation}
        \Delta(x,y)
        =
        \{\mathcal C_0(x),\mathcal C_1(y)\}
        =
        \mathfrak d(x)\delta^{(2)}(x-y).
        \label{eq:Delta-local-form-vacuum-section}
\end{equation}
We therefore do not need the full nonlinear expression for
\(\mathfrak d\).  It is enough to evaluate it at a regular point of phase
space.  If it is nonzero at the maximally symmetric vacuum, continuity
guarantees an open neighborhood of that vacuum on which the scalar Dirac
matrix has the same rank.

Recall that
\begin{equation}
        \mathcal C_1
        =
        \{\mathcal C_0,H_0[1]\}.
        \label{eq:C1-unit-lapse-vacuum}
\end{equation}
The unit lapse in this equation is only a convenient way of defining the
local phase-space function \(\mathcal C_1\).  It is not a gauge choice:
Proposition~\ref{prop:arbitrary-lapse-factorization} established the exact
identity
\[
        \{\mathcal C_0(x),H_0[N]\}
        =
        N(x)\mathcal C_1(x)
\]
for arbitrary lapse \(N(x)\).

We evaluate the bracket on a time-symmetric slice of the maximally
symmetric BINMG vacuum and use spatial Riemann normal coordinates at the
point under consideration.  The auxiliary vacuum relation is
\begin{equation}
        \bar q_{\mu\nu}
        =
        a\,\bar g_{\mu\nu},
        \qquad
        a=\beta^2,
        \qquad
        \beta>0 .
        \label{eq:q-vacuum-a}
\end{equation}
At the chosen point the background canonical data can be written as
\begin{equation}
\begin{gathered}
        \bar\gamma_{ij}=\delta_{ij},
        \qquad
        \bar S^{ij}=\beta\delta^{ij},
        \qquad
        \bar p_{ij}=0,
        \qquad
        \bar\rho^{ij}=0,
        \\
        \bar s=\beta^2=a,
        \qquad
        {}^{(2)}\bar R=2\Lambda .
\end{gathered}
        \label{eq:vacuum-local-canonical-data}
\end{equation}
Covariantly,
\(\bar S^{ij}=\beta\sqrt{\bar\gamma}\,\bar\gamma^{ij}\).
The scalar constraint is satisfied \eqref{eq:scalar-constraint-vac}.

\paragraph{Why the linearized constraints are sufficient.}

To evaluate the full nonlinear Poisson bracket at a fixed background, only
the first variations of the two functions are required.  Indeed, if
\(\chi_A\) and \(\chi_B\) are phase-space functions, then
\begin{equation}
        \left.
        \{\chi_A(x),\chi_B(y)\}
        \right|_{\bar z}
        =
        \{
        \chi_A^{(1)}(x),
        \chi_B^{(1)}(y)
        \}_{\rm lin}.
        \label{eq:background-rank-identity}
\end{equation}
Thus the calculation below is an evaluation of the full nonlinear Dirac
matrix at the vacuum; it is not an inference of the nonlinear
degree-of-freedom count from the linearized spectrum.

Introduce the canonical perturbations
\begin{equation}
        h_{ij}:=\delta\gamma_{ij},
        \qquad
        t^{ij}:=\delta S^{ij},
        \qquad
        r^{ij}:=\delta\rho^{ij},
        \qquad
        \ell_{ij}:=\delta p_{ij},
        \label{eq:vacuum-canonical-perturbations}
\end{equation}
with traces
\begin{equation}
        h:=\bar\gamma^{ij}h_{ij},
        \qquad
        t:=\bar\gamma_{ij}t^{ij},
        \qquad
        r:=\bar\gamma_{ij}r^{ij},
        \qquad
        \ell:=\bar\gamma^{ij}\ell_{ij}.
        \label{eq:vacuum-perturbation-traces}
\end{equation}

\paragraph{Linearized form of \(\mathcal C_0\).}

At the vacuum,
\begin{equation}
        \delta s
        =
        \bar s(\bar S^{-1})_{ij}\delta S^{ij}
        =
        \beta t .
        \label{eq:delta-s-vacuum}
\end{equation}
In two spatial dimensions,
\begin{equation}
        \delta({}^{(2)}R)
        =
        D_iD_jh^{ij}
        -
        D^2h
        -
        \Lambda h .
        \label{eq:R2-linear-vacuum}
\end{equation}
Therefore
\begin{equation}
        \mathcal C_0^{(1)}
        =
        -\frac{\beta}{2}t
        -
        \frac c2
        \left(
        D_iD_jh^{ij}
        -
        D^2h
        -
        \Lambda h
        \right).
        \label{eq:C0-linear-vacuum-full}
\end{equation}

We shall also need the momentum-dependent quadratic part,
\begin{equation}
        \left.
        \mathcal C_0^{(2)}
        \right|_{\rm mom}
        =
        -\frac1{2c}
        \left(
        \ell^2-\ell_{ij}\ell^{ij}
        \right).
        \label{eq:C0-quadratic-momentum}
\end{equation}

\paragraph{Linearized form of \(\mathcal C_1\).}

Since
\begin{equation}
        \mathcal C_1
        =
        \{\mathcal C_0,H_0[1]\},
\end{equation}
its linear part is
\begin{equation}
        \mathcal C_1^{(1)}
        =
        \{\mathcal C_0^{(1)},H_0^{(2)}[1]\}
        +
        \{\mathcal C_0^{(2)},H_0^{(1)}[1]\}.
        \label{eq:C1-linear-unit-lapse-split}
\end{equation}

To first order in the canonical momenta,
\begin{equation}
        \Pi_{ij}
        =
        \ell_{ij}-\ell\bar\gamma_{ij},
        \qquad
        \pi^{ij}
        =
        r^{ij}
        -
        2\beta\ell^{ij}
        +
        \beta\ell\bar\gamma^{ij}.
        \label{eq:vacuum-linear-Pi-pi}
\end{equation}
At the same background,
\begin{equation}
        \mathcal Q(S,\Pi)
        =
        \beta\Pi_{ij}\Pi^{ij}
        =
        \beta\ell_{ij}\ell^{ij}.
        \label{eq:vacuum-linear-Q}
\end{equation}
The complete momentum-dependent quadratic part of the unit-lapse
Hamiltonian density is therefore
\begin{equation}
\begin{aligned}
        \mathcal H_{0,{\rm mom}}^{(2)}
        ={}&
        -\frac2c
        \left(
        r^{ij}\ell_{ij}
        -
        r\ell
        \right)
        +
        \frac{3\beta}{c}\ell_{ij}\ell^{ij}
        -
        \frac{2\beta}{c}\ell^2
        \\
        &-
        \frac2\beta
        \left(D_j\ell^j{}_i\right)
        \left(D_k\ell^{ki}\right).
\end{aligned}
        \label{eq:H0-momentum-quadratic-vacuum}
\end{equation}

For unit lapse, the spatial integral of
\(cD_iD_j\mathcal A^{ij}\) is a boundary term.  The bulk part of the
Hamiltonian linear in the configuration perturbations needed below is
\begin{equation}
        H_0^{(1)}[1]
        =
        -\frac12
        \int_{\Sigma_t}\dd^2x\,t ,
        \label{eq:H0-linear-vacuum}
\end{equation}
up to the boundary terms excluded in the local analysis.

The three contributions to \(\mathcal C_1^{(1)}\) can now be displayed
separately.  The curvature part of \(\mathcal C_0^{(1)}\), bracketed with
the \(r^{ij}\ell_{ij}-r\ell\) term in
Eq.~\eqref{eq:H0-momentum-quadratic-vacuum}, gives
\begin{equation}
        D_iD_j\ell^{ij}
        +
        \Lambda\ell .
        \label{eq:C1-vacuum-contribution-one}
\end{equation}
The term \(-\beta t/2\) in \(\mathcal C_0^{(1)}\), bracketed with the
complete \(\ell_{ij}\)-dependence of
Eq.~\eqref{eq:H0-momentum-quadratic-vacuum}, gives
\begin{equation}
        -2D_iD_j\ell^{ij}
        -
        \frac{\beta}{c}r
        +
        \frac{\beta^2}{c}\ell .
        \label{eq:C1-vacuum-contribution-two}
\end{equation}
Finally,
\(\left.\mathcal C_0^{(2)}\right|_{\rm mom}\), bracketed with
Eq.~\eqref{eq:H0-linear-vacuum}, gives
\begin{equation}
        -\frac1{2c}\ell .
        \label{eq:C1-vacuum-contribution-three}
\end{equation}
Adding the three terms yields
\begin{equation}
        \mathcal C_1^{(1)}
        =
        -D_iD_j\ell^{ij}
        -
        \frac{\beta}{c}r
        +
        \alpha_\Lambda\ell ,
        \label{eq:C1-linear-vacuum-full}
\end{equation}
where
\begin{equation}
        \alpha_\Lambda
        :=
        \Lambda
        +
        \frac{a-\tfrac12}{c}.
        \label{eq:alpha-Lambda-definition}
\end{equation}

\paragraph{The scalar bracket at the vacuum.}

By Eq.~\eqref{eq:background-rank-identity},
\begin{equation}
        \Delta_{\rm vac}(x,y)
        =
        \{
        \mathcal C_0^{(1)}(x),
        \mathcal C_1^{(1)}(y)
        \}_{\rm lin}.
        \label{eq:Delta-vac-linear-bracket}
\end{equation}
The nonzero canonical brackets are
\begin{equation}
        \{h_{ij}(x),r^{kl}(y)\}
        =
        \delta_{(i}{}^k\delta_{j)}{}^l
        \delta^{(2)}(x-y),
        \label{eq:h-r-bracket-vac}
\end{equation}
and
\begin{equation}
        \{t^{ij}(x),\ell_{kl}(y)\}
        =
        \delta_{(k}{}^i\delta_{l)}{}^j
        \delta^{(2)}(x-y).
        \label{eq:t-l-bracket-vac}
\end{equation}
In particular,
\begin{equation}
        \{h(x),r(y)\}
        =
        2\delta^{(2)}(x-y),
        \qquad
        \{t(x),\ell(y)\}
        =
        2\delta^{(2)}(x-y).
        \label{eq:vacuum-trace-brackets}
\end{equation}

There are only two contributions.  The \(h\)--\(r\) bracket gives
\begin{equation}
\begin{aligned}
 &\left\{
 -\frac c2
 \left(
 D_iD_jh^{ij}-D^2h-\Lambda h
 \right)(x),
 -\frac{\beta}{c}r(y)
 \right\}
 =
 -\frac{\beta}{2}D^2\delta^{(2)}(x-y)
 -
 \beta\Lambda\,\delta^{(2)}(x-y).
\end{aligned}
        \label{eq:vacuum-bracket-hr}
\end{equation}
The \(t\)--\(\ell\) bracket gives
\begin{equation}
\begin{aligned}
 &\left\{
 -\frac{\beta}{2}t(x),
 -D_iD_j\ell^{ij}(y)
 +
 \alpha_\Lambda\ell(y)
 \right\}
 =
 +\frac{\beta}{2}D^2\delta^{(2)}(x-y)
 -
 \beta\alpha_\Lambda\,
 \delta^{(2)}(x-y).
\end{aligned}
        \label{eq:vacuum-bracket-tl}
\end{equation}

The derivative-of-delta terms cancel exactly, in agreement with
Proposition~\ref{prop:Delta-exact-ultralocal}.  What remains is
\begin{align}
        \Delta_{\rm vac}(x,y)
        &=
        -\beta
        \left(
        \Lambda+\alpha_\Lambda
        \right)
        \delta^{(2)}(x-y)=
        -\frac{\beta}{c}
        \left(
        2c\Lambda+a-\frac12
        \right)
        \delta^{(2)}(x-y).
        \label{eq:Delta-vacuum-before-vac-relation}
\end{align}
Using the vacuum relation
\(
2c\Lambda=1-a
\),
we find
\begin{equation}
        \Delta_{\rm vac}(x,y)
        =
        -\frac{\beta}{2c}
        \delta^{(2)}(x-y)
        =
        -\frac{\sqrt a}{2c}
        \delta^{(2)}(x-y).
        \label{eq:Delta-vacuum-kernel}
\end{equation}
Equivalently,
\begin{equation}
        \mathfrak d_{\rm vac}
        =
        -\frac{\beta}{2c}
        =
        -\frac{\sqrt a}{2c}.
        \label{eq:Delta-vacuum-direct}
\end{equation}
Since \(a>0\) and \(c\neq0\),
\begin{equation}
        \mathfrak d_{\rm vac}\neq0 .
\end{equation}

\begin{proposition}[Scalar rank at the BINMG vacuum]
\label{prop:vacuum-rank}
At the nondegenerate maximally symmetric BINMG vacuum,
\begin{equation}
        \Delta_{\rm vac}(x,y)
        =
        -\frac{\sqrt a}{2c}
        \delta^{(2)}(x-y),
\end{equation}
and therefore
\(\mathcal C_0\) and \(\mathcal C_1\) form a second-class pair.
\end{proposition}

\begin{proof}
Since
\begin{equation}
        \mathfrak d_{\rm vac}
        =
        -\frac{\sqrt a}{2c}
        \neq0,
\end{equation}
the off-diagonal kernel is invertible, with
\begin{equation}
        \Delta_{\rm vac}^{-1}(x,y)
        =
        -\frac{2c}{\sqrt a}
        \delta^{(2)}(x-y).
\end{equation}
The block argument of Sec.~\ref{sec:scalar-dirac} then implies that the
full scalar Dirac matrix is invertible.
\end{proof}

\paragraph{The regular vacuum-connected component.}

Because \(\mathfrak d\) is a local function of the canonical fields and a
finite number of their spatial derivatives, the nonzero value
\(\mathfrak d_{\rm vac}\) implies the existence of an open neighborhood
of the vacuum on which \(\mathfrak d\neq0\).

We now define the region on which the nonlinear degree-of-freedom theorem
will be stated.  Let \(\mathcal U_{\rm vac}\) be the connected component
containing the maximally symmetric BINMG vacuum of the open regular region
on which
\begin{equation}
        s\neq0,
        \qquad
        \det q_{\mu\nu}\neq0,
        \qquad
        \mathfrak d(x)\neq0 .
        \label{eq:Uvac-regularity-conditions}
\end{equation}
The three conditions have distinct origins.  The condition \(s\neq0\)
allows the algebraic elimination of \(v^i\);
\(\det q_{\mu\nu}\neq0\) is required for the auxiliary-metric formulation;
and \(\mathfrak d\neq0\) is precisely the scalar-rank condition.

On \(\mathcal U_{\rm vac}\),
\begin{equation}
        \mathcal C_0,
        \qquad
        \mathcal C_1
        \quad\hbox{form a second-class pair}.
\end{equation}
No claim is made here about configurations at which
\begin{equation}
        s=0,
        \qquad
        \det q_{\mu\nu}=0,
        \qquad
        \hbox{or}
        \qquad
        \mathfrak d=0 .
\end{equation}
At such loci the algebraic structure or the rank of the Dirac system can
change, and those cases require a separate analysis.

Finally, the nonvanishing mutual bracket also proves that
\(\mathcal C_1\) is independent of the previously existing constraints.
If \(\mathcal C_1\) were weakly a regular linear combination of those
constraints, then its bracket with \(\mathcal C_0\) would vanish weakly:
the self-bracket of \(\mathcal C_0\) vanishes strongly, its bracket with
the tangential first-class generators is proportional to
\(\mathcal C_0\), and its bracket with the normal first-class
representative is weakly proportional to
\[
        \{\mathcal C_0,H_0[1]\}
        =
        \mathcal C_1
        \weak0 .
\]
This would imply
\begin{equation}
        \{\mathcal C_0,\mathcal C_1\}\weak0,
\end{equation}
contradicting
\begin{equation}
        \{\mathcal C_0(x),\mathcal C_1(y)\}
        =
        \mathfrak d(x)\delta^{(2)}(x-y),
        \qquad
        \mathfrak d\neq0
\end{equation}
on \(\mathcal U_{\rm vac}\).  Thus \(\mathcal C_1\) is a genuinely
independent secondary constraint.

We have therefore completed the nonlinear scalar-rank analysis.  On
\(\mathcal U_{\rm vac}\),
\begin{equation}
        \mathcal C_0
        \ \longrightarrow\
        \mathcal C_1,
        \qquad
        \{\mathcal C_0(x),\mathcal C_1(y)\}
        =
        \mathfrak d(x)\delta^{(2)}(x-y),
        \qquad
        \mathfrak d\neq0 .
\end{equation}
Hence \((\mathcal C_0,\mathcal C_1)\) is a second-class pair.
Preservation of \(\mathcal C_1\) fixes the auxiliary variable \(\nu\),
not the physical lapse, and no further scalar constraint is generated.
The scalar sector therefore removes exactly two phase-space directions.

\subsection{Reduced and extended bookkeeping}
\label{sec:reduced-extended-bookkeeping}

The nonlinear calculation above was carried out after eliminating the
algebraic variable \(v^i\), while \(N\), \(N^i\), and \(\nu\) were treated
as nondynamical variables.  This is the most economical formulation for
the actual constraint analysis.  For the final degree-of-freedom count,
however, it is useful to compare it with the fully extended Dirac phase
space, in which all of these variables and their conjugate momenta are
restored.

No new constraint is derived in this subsection.  We are only restoring
the canonical pairs that were suppressed or eliminated and keeping track
of the corresponding first- and second-class constraints.

\paragraph{Reduced dynamical formulation.}

The dynamical canonical phase space used above is
\begin{equation}
        (\gamma_{ij},\rho^{ij};
        S^{ij},p_{ij}).
\end{equation}
Each symmetric tensor in two spatial dimensions has three independent
components, so this phase space has
\begin{equation}
        N_{\rm phase}^{\rm red}=12
\end{equation}
dimensions per spatial point.

On \(\mathcal U_{\rm vac}\), the scalar sector contains the second-class
pair
\begin{equation}
        \mathcal C_0\weak0,
        \qquad
        \mathcal C_1\weak0,
        \label{eq:reduced-scalar-pair-bookkeeping}
\end{equation}
with
\begin{equation}
        \{\mathcal C_0(x),\mathcal C_1(y)\}
        =
        \mathfrak d(x)\delta^{(2)}(x-y),
        \qquad
        \mathfrak d\neq0 .
        \label{eq:reduced-scalar-bracket-bookkeeping}
\end{equation}
The same reduced phase space carries three first-class generators of
spacetime diffeomorphisms: one normal and two tangential.  Hence
\begin{equation}
        N_{\rm dof}
        =
        \frac12
        \left(
        12
        -
        2\times3
        -
        2
        \right)
        =
        2 .
        \label{eq:reduced-dof-check}
\end{equation}

We now recover the same result without eliminating the auxiliary
canonical pairs.

\paragraph{Restoring the pair \((\nu,\pi_\nu)\).}

First restore the canonical pair
\begin{equation}
        (\nu,\pi_\nu),
\end{equation}
while continuing to regard \(N\) and \(N^i\) as prescribed gauge
functions.  The corresponding total Hamiltonian is
\begin{equation}
        H_T^{(\nu)}
        =
        H_0[N]
        +
        H_{\rm sp}[N^i]
        +
        \mathcal C_0[\nu]
        +
        \pi_\nu[u_\nu].
        \label{eq:HT-normal-chain}
\end{equation}
Since the action contains no \(\dot\nu\),
\begin{equation}
        \pi_\nu\weak0
\end{equation}
is a primary constraint.  Its preservation gives
\begin{equation}
        \dot\pi_\nu(x)
        =
        \{\pi_\nu(x),H_T^{(\nu)}\}
        =
        -\mathcal C_0(x)
        \weak0,
        \label{eq:pi-nu-gives-C0}
\end{equation}
and hence
\begin{equation}
        \mathcal C_0\weak0 .
\end{equation}

The next consistency condition is
\begin{equation}
        \dot{\mathcal C}_0(x)
        \weak
        N(x)\mathcal C_1(x),
        \label{eq:C0-gives-C1-extended}
\end{equation}
where the tangential term is proportional to the already imposed
constraint.  On a regular ADM chart \(N\neq0\), this produces
\begin{equation}
        \mathcal C_1\weak0 .
\end{equation}

Preservation of \(\mathcal C_1\) gives
\begin{equation}
\begin{aligned}
        0
        \weak
        \dot{\mathcal C}_1(x)
        ={}&
        F_N(x)
        +
        \int_{\Sigma_t}\dd^2y\,
        \nu(y)
        \{\mathcal C_1(x),\mathcal C_0(y)\},
\end{aligned}
        \label{eq:C1-preservation-bookkeeping}
\end{equation}
where
\begin{equation}
        F_N(x)
        :=
        \{\mathcal C_1(x),H_0[N]\}.
        \label{eq:FN-definition-normal-chain}
\end{equation}
The spatial-diffeomorphism contribution has again been omitted because it
is proportional to \(\mathcal C_1\).

Using
\begin{equation}
        \{\mathcal C_1(x),\mathcal C_0(y)\}
        =
        -\mathfrak d(x)\delta^{(2)}(x-y),
\end{equation}
we obtain
\begin{equation}
        0
        \weak
        F_N(x)-\mathfrak d(x)\nu(x).
        \label{eq:C1-fixes-nu-bookkeeping}
\end{equation}
Since \(\mathfrak d\neq0\) on \(\mathcal U_{\rm vac}\), this determines
\(\nu\):
\begin{equation}
        \chi_\nu(x)
        :=
        \nu(x)-\nu_*(x)
        \weak0,
        \qquad
        \nu_*(x)
        :=
        \frac{F_N(x)}{\mathfrak d(x)} .
        \label{eq:chi-nu-pointwise}
\end{equation}

Thus the normal auxiliary chain is
\begin{equation}
        \pi_\nu
        \ \longrightarrow\
        \mathcal C_0
        \ \longrightarrow\
        \mathcal C_1
        \ \longrightarrow\
        \chi_\nu .
        \label{eq:nu-Dirac-chain-four}
\end{equation}
The last condition fixes the auxiliary coordinate \(\nu\); it does not
restrict the physical lapse \(N\).  Preservation of \(\chi_\nu\) then
fixes the multiplier \(u_\nu\), so the chain terminates.

The four constraints
\begin{equation}
        \pi_\nu,
        \qquad
        \chi_\nu,
        \qquad
        \mathcal C_0,
        \qquad
        \mathcal C_1
        \label{eq:normal-four-constraints}
\end{equation}
form a second-class block.  To see this directly, order them as
\begin{equation}
        \Phi_A
        =
        (\pi_\nu,\chi_\nu,\mathcal C_0,\mathcal C_1).
\end{equation}
Their weak Poisson matrix has the schematic form
\begin{equation}
        \Omega^{(\nu)}_{AB}
        \weak
        \begin{pmatrix}
        0 & -I & 0 & 0 \\
        I & 0 & A & B \\
        0 & -A^T & 0 & \Delta \\
        0 & -B^T & -\Delta^T & \Gamma
        \end{pmatrix},
        \label{eq:normal-four-by-four-Dirac-matrix}
\end{equation}
where
\begin{equation}
        \Delta(x,y)
        =
        \mathfrak d(x)\delta^{(2)}(x-y)
\end{equation}
is invertible on \(\mathcal U_{\rm vac}\).

Indeed, if
\begin{equation}
        \Omega^{(\nu)}
        \begin{pmatrix}
        a\\ b\\ c\\ d
        \end{pmatrix}
        =0,
\end{equation}
the first row gives \(b=0\); the third then gives
\(\Delta d=0\), hence \(d=0\); the fourth gives
\(\Delta^Tc=0\), hence \(c=0\); and finally the second gives
\(a=0\).  Thus the normal auxiliary sector contributes exactly four
second-class constraints.

\paragraph{Restoring the mixed auxiliary pair \((v^i,\pi_i^{(v)})\).}

Before the algebraic elimination carried out in
Sec.~\ref{sec:auxiliary-shift-elimination}, the mixed auxiliary component
belongs to the canonical pair
\begin{equation}
        (v^i,\pi_i^{(v)}).
\end{equation}
Since no \(\dot v^i\) occurs, there are two primary constraints
\begin{equation}
        \pi_i^{(v)}\weak0 .
\end{equation}
Their preservation gives the two algebraic conditions
\begin{equation}
        \chi_{(v)}^i
        :=
        v^i-v_*^i
        \weak0,
        \qquad
        v_*^i
        =
        -\frac2sS^{ij}D_kp^k{}_j .
        \label{eq:v-vector-second-class-block}
\end{equation}
Because
\begin{equation}
        \{\pi_i^{(v)}(x),\chi_{(v)}^j(y)\}
        =
        -\delta_i{}^j\delta^{(2)}(x-y),
        \label{eq:v-block-identity}
\end{equation}
these four conditions form another invertible second-class block.

It is also useful to verify that eliminating this block before performing
the scalar calculation does not alter the Poisson brackets used there.
In the ordering
\begin{equation}
        \bigl(\pi_i^{(v)},\chi_{(v)}^i\bigr),
\end{equation}
the vector Dirac matrix has the form
\begin{equation}
        \Omega^{(v)}
        =
        \begin{pmatrix}
        0 & -I \\
        I & C
        \end{pmatrix},
\end{equation}
with inverse
\begin{equation}
        \bigl(\Omega^{(v)}\bigr)^{-1}
        =
        \begin{pmatrix}
        C & I \\
        -I & 0
        \end{pmatrix}.
\end{equation}
For functionals \(F\) and \(G\) independent of
\(v^i\) and \(\pi_i^{(v)}\),
\begin{equation}
        \{F,\pi_i^{(v)}\}
        =
        \{G,\pi_i^{(v)}\}
        =
        0 .
\end{equation}
Because the lower-right block of
\(\bigl(\Omega^{(v)}\bigr)^{-1}\) vanishes, the Dirac-bracket correction
also vanishes:
\begin{equation}
        \{F,G\}_{D(v)}
        =
        \{F,G\}.
\end{equation}
Thus eliminating \(v^i\) first and then using ordinary Poisson brackets
for the remaining canonical variables is legitimate.

The two auxiliary sectors together therefore contribute
\begin{equation}
        4+4=8
\end{equation}
second-class constraints,
\begin{equation}
        \pi_i^{(v)},
        \quad
        \chi_{(v)}^i,
        \quad
        \pi_\nu,
        \quad
        \chi_\nu,
        \quad
        \mathcal C_0,
        \quad
        \mathcal C_1 ,
        \label{eq:eight-auxiliary-second-class}
\end{equation}
where each vector quantity counts twice.

\paragraph{Fully extended phase space.}

Finally restore the lapse and shift as canonical coordinates,
\begin{equation}
        (N,\pi_N),
        \qquad
        (N^i,\pi_i).
\end{equation}
The complete configuration space is then
\begin{equation}
        \gamma_{ij},
        \qquad
        S^{ij},
        \qquad
        N,
        \qquad
        N^i,
        \qquad
        \nu,
        \qquad
        v^i .
\end{equation}
Counting components in two spatial dimensions gives
\begin{equation}
        3+3+1+2+1+2=12
\end{equation}
configuration variables, or
\begin{equation}
        N_{\rm phase}^{\rm ext}=24
\end{equation}
phase-space dimensions per spatial point.

The variables \(N,N^i,\nu,v^i\) carry no independent velocities and give
the six primary constraints
\begin{equation}
        \pi_N,
        \qquad
        \pi_i,
        \qquad
        \pi_\nu,
        \qquad
        \pi_i^{(v)} .
        \label{eq:six-primary-constraints-full}
\end{equation}
The invertible six-velocity Hessian of
Sec.~\ref{sec:velocity-block-explained} shows that there are no additional
primary constraints in the dynamical
\((\gamma_{ij},S^{ij})\) sector.

The auxiliary sector contains the eight second-class constraints displayed
in Eq.~\eqref{eq:eight-auxiliary-second-class}.  The diffeomorphism sector
contains three gauge functions, each entering with one time derivative,
and therefore three two-stage first-class chains,
\begin{equation}
        \Phi_\alpha^\star
        \ \longrightarrow\
        \mathcal H_\alpha^\star,
        \qquad
        \alpha=\perp,1,2 .
        \label{eq:projected-diffeomorphism-chains}
\end{equation}
The stars indicate that the normal primary and secondary representatives
may have to be projected away from the auxiliary second-class sector, as
explained in Sec.~\ref{sec:diffeomorphism-generators}.  This changes the
representatives but not the number of gauge directions.

The fully extended theory therefore contains
\begin{equation}
        3\ {\rm primary}
        +
        3\ {\rm secondary}
        =
        6\ {\rm first\!-\!class\ constraints}
        \label{eq:six-first-class-final}
\end{equation}
and
\begin{equation}
        8\ {\rm second\!-\!class\ constraints}.
\end{equation}

The fully extended Dirac count is consequently
\begin{equation}
        N_{\rm dof}
        =
        \frac12
        \left(
        24
        -
        2\times6
        -
        8
        \right)
        =
        2 .
        \label{eq:extended-dof-check}
\end{equation}
This is exactly the same physical count as in the reduced formulation,
Eq.~\eqref{eq:reduced-dof-check}.

\subsection{The full degree-of-freedom count at a glance}
\label{sec:dof-count-at-a-glance}

We collect the count in one place.  The configuration variables are
\begin{equation}
        \gamma_{ij},
        \qquad
        S^{ij},
        \qquad
        N,
        \qquad
        N^i,
        \qquad
        \nu,
        \qquad
        v^i .
\end{equation}
Since a symmetric tensor in two spatial dimensions has three independent
components, the corresponding phase-space dimensions are
\begin{equation}
\begin{aligned}
24
={}
\underbrace{6}_{(\gamma_{ij},\rho^{ij})}
+
\underbrace{6}_{(S^{ij},p_{ij})}
+
\underbrace{2}_{(N,\pi_N)}
+
\underbrace{4}_{(N^i,\pi_i)}
+
\underbrace{2}_{(\nu,\pi_\nu)}
+
\underbrace{4}_{(v^i,\pi_i^{(v)})}.
\end{aligned}
\label{eq:full-phase-space-summary}
\end{equation}

Three constraint sectors act on these.  First, diffeomorphism invariance
gives three two-stage first-class chains,
\begin{equation}
        \Phi_\alpha^\star
        \ \longrightarrow\
        \mathcal H_\alpha^\star,
        \qquad
        \alpha=\perp,1,2 ,
\end{equation}
and hence
\begin{equation}
        3+3=6
        \qquad\text{first-class constraints}.
\end{equation}
They remove
\begin{equation}
        2\times6=12
\end{equation}
phase-space directions.

Second, the mixed auxiliary component gives the four second-class
constraints
\begin{equation}
        \pi_i^{(v)}\weak0,
        \qquad
        \chi_{(v)}^i
        :=
        v^i-v_*^i
        \weak0,
        \qquad
        i=1,2 ,
\end{equation}
which remove the complete four-dimensional canonical sector associated
with \(v^i\).

Third, the auxiliary normal sector gives the chain
\begin{equation}
        \pi_\nu
        \ \longrightarrow\
        \mathcal C_0
        \ \longrightarrow\
        \mathcal C_1
        \ \longrightarrow\
        \chi_\nu ,
\label{eq:normal-chain-summary}
\end{equation}
where
\begin{equation}
        \chi_\nu
        :=
        \nu-\nu_* .
\end{equation}
The nonlinear input is
\begin{equation}
        \{\mathcal C_0(x),\mathcal C_1(y)\}
        =
        \mathfrak d(x)\delta^{(2)}(x-y),
        \qquad
        \mathfrak d\neq0
        \quad\text{on }\mathcal U_{\rm vac}.
\label{eq:scalar-rank-summary}
\end{equation}
Consequently
\begin{equation}
        \pi_\nu,\qquad
        \chi_\nu,\qquad
        \mathcal C_0,\qquad
        \mathcal C_1
\end{equation}
form another four-dimensional second-class block.

Thus the complete constraint content is
\begin{equation}
        N_{\rm phase}=24,
        \qquad
        N_{\rm FC}=6,
        \qquad
        N_{\rm SC}=4+4=8 .
\label{eq:full-count-summary}
\end{equation}
Therefore
\begin{equation}
\begin{aligned}
N_{\rm phase}^{\rm phys}
&=
24
-
2\times6
-
8
\\
&=
4,
\end{aligned}
\qquad
N_{\rm dof}
=
\frac{4}{2}
=
2 .
\label{eq:dof-summary-final}
\end{equation}

For bookkeeping purposes the same reduction may be pictured as
\begin{equation}
24
\ \xrightarrow{\;4\ {\rm SC}\;(v^i)\;}
20
\ \xrightarrow{\;4\ {\rm SC}\;(\nu)\;}
16
\ \xrightarrow{\;6\ {\rm FC}\;}
4
\ \xrightarrow{\;/2\;}
2\ {\rm local\ degrees\ of\ freedom}.
\label{eq:dof-count-flow}
\end{equation}
The order of the arrows is for display only; it is not the order of the
Dirac algorithm.

\section{Degree-of-freedom theorem}
\label{sec:degree-of-freedom-theorem}

We can now state the main result.

Recall that \(\mathcal U_{\rm vac}\) denotes the connected component
containing the maximally symmetric BINMG vacuum of the regular region
defined by
\begin{equation}
        s=\det S^{ij}\neq0,
        \qquad
        \det q_{\mu\nu}\neq0,
        \qquad
        \mathfrak d\neq0 .
        \label{eq:theorem-regularity-recall}
\end{equation}
The first condition allows the algebraic elimination of \(v^i\), the
second guarantees the nondegeneracy required by the auxiliary-metric
formulation, and the third states that the scalar Dirac matrix has
nonzero rank.  In addition,
\begin{equation}
        c=\frac{\sigma}{2m^2}\neq0
\end{equation}
for finite \(m^2\) and \(\sigma=\pm1\), so the Legendre map of the
dynamical sector is nonsingular.

\begin{theorem}[Nonlinear degrees of freedom of BINMG]
\label{thm:BINMG-two-dof}
On the regular vacuum-connected component
\(\mathcal U_{\rm vac}\), Born--Infeld New Massive Gravity has exactly two
local physical configuration-space degrees of freedom per spatial point.
\end{theorem}

\begin{proof}
The Hamiltonian analysis established the following three ingredients.

First, the velocity Hessian of the dynamical variables
\((\gamma_{ij},S^{ij})\) is invertible for \(c\neq0\).  Hence there are no
additional primary constraints in this sector.

Second, spacetime diffeomorphism invariance gives three independent gauge
functions.  In the fully extended phase space they generate three
two-stage first-class chains and therefore six first-class constraints.

Third, the auxiliary sector contains eight second-class constraints.
Four arise from the two components of \(v^i\) and their primary momenta.
The remaining four arise from the normal auxiliary chain
\begin{equation}
        \pi_\nu
        \ \longrightarrow\
        \mathcal C_0
        \ \longrightarrow\
        \mathcal C_1
        \ \longrightarrow\
        \chi_\nu .
\end{equation}
The nontrivial step is the scalar pair
\((\mathcal C_0,\mathcal C_1)\), whose mutual bracket is
\begin{equation}
        \{\mathcal C_0(x),\mathcal C_1(y)\}
        =
        \mathfrak d(x)\delta^{(2)}(x-y),
\end{equation}
with \(\mathfrak d\neq0\) throughout \(\mathcal U_{\rm vac}\).

The fully extended phase space has twenty-four dimensions per spatial
point.  The physical phase-space dimension is therefore
\begin{equation}
        N_{\rm phase}^{\rm phys}
        =
        24
        -
        2\times6
        -
        8
        =
        4 .
\end{equation}
Hence
\begin{equation}
        N_{\rm dof}
        =
        \frac12 N_{\rm phase}^{\rm phys}
        =
        2 .
\end{equation}
\end{proof}

The nonlinear count therefore agrees with the two local polarizations of a
massive spin-two field in three dimensions.  The essential nonlinear
mechanism is the second-class scalar pair
\((\mathcal C_0,\mathcal C_1)\), which removes one canonical pair that
would otherwise represent an additional scalar degree of freedom.

\subsection{Hamiltonian interpretation}
\label{sec:hamiltonian-of-theory}

It is useful to clarify one point that can otherwise look paradoxical.
BINMG possesses local physical degrees of freedom, yet after the auxiliary
second-class sector has been eliminated its bulk Hamiltonian is a linear
combination of diffeomorphism constraints.

There is no contradiction.  These statements concern different aspects of
the theory.  The number of local degrees of freedom is determined by the
dimension of the reduced physical phase space, whereas the form of the
Hamiltonian reflects spacetime diffeomorphism invariance.

After eliminating the second-class constraints, either explicitly or by
using the corresponding Dirac bracket, the bulk Hamiltonian may be written
as
\begin{equation}
        H_{\rm bulk}[N,N^i]
        =
        \int_{\Sigma_t}\dd^2x
        \left(
        N\mathcal H_\perp^\star
        +
        N^i\mathcal H_i
        \right),
        \label{eq:bulk-Hamiltonian-final}
\end{equation}
where \(\mathcal H_i\) is the tangential diffeomorphism generator found in
Sec.~\ref{sec:diffeomorphism-generators}, while
\(\mathcal H_\perp^\star\) is an appropriate first-class representative of
the normal generator.  On the auxiliary second-class surface,
\begin{equation}
        \mathcal H_\perp^\star
        \weak
        \mathcal H_0 .
\end{equation}

If the lapse and shift are themselves retained as canonical variables, the
total Hamiltonian also contains the primary first-class constraints,
\begin{equation}
        H_T
        =
        H_{\rm bulk}
        +
        \int_{\Sigma_t}\dd^2x\,
        u^\alpha\Phi_\alpha^\star,
        \qquad
        \alpha=\perp,1,2 .
        \label{eq:extended-total-Hamiltonian-final}
\end{equation}

For a compact spatial slice without boundary,
\begin{equation}
        H_{\rm bulk}[N,N^i]\weak0 .
        \label{eq:bulk-H-weak-zero}
\end{equation}
This is the familiar Hamiltonian manifestation of general covariance.  It
does not imply the absence of local physical degrees of freedom: the latter
is determined by the constraint count established above.

If the spatial slice has a boundary or an asymptotic end, the bulk
generator generally fails to be functionally differentiable by itself and
must be supplemented by an appropriate boundary functional,
\begin{equation}
        H[N,N^i]
        =
        H_{\rm bulk}[N,N^i]
        +
        H_{\partial\Sigma}[N,N^i].
        \label{eq:Hamiltonian-with-boundary-final}
\end{equation}
On the full bulk constraint surface,
\begin{equation}
        H[N,N^i]
        \weak
        H_{\partial\Sigma}[N,N^i].
\end{equation}
For boundary conditions admitting asymptotic time translations and
rotations, the corresponding boundary terms give the conserved energy and
angular momentum, respectively \cite{ADM,DeserTekin}.  Their detailed form
depends on the boundary conditions and plays no role in the local
degree-of-freedom count.

\section{Domain of validity and relation to previous results}
\label{sec:scope-and-discussion}

Theorem~\ref{thm:BINMG-two-dof} is a local statement in phase space.  Its
domain is the vacuum-connected regular component
\(\mathcal U_{\rm vac}\), characterized by
\begin{equation}
        s\neq0,
        \qquad
        \det q_{\mu\nu}\neq0,
        \qquad
        \mathfrak d\neq0 .
        \label{eq:scope-Uvac}
\end{equation}
These conditions have distinct origins.  The first is required to solve
algebraically for \(v^i\); the second is required for the auxiliary metric
\(q_{\mu\nu}\) to be nondegenerate; and the third is the nonzero-rank
condition for the scalar Dirac matrix.  The spatial metric is assumed
nondegenerate throughout, as required by the ADM decomposition.

The condition
\begin{equation}
        N\neq0
\end{equation}
has a different status.  It is only the regularity condition for the ADM
foliation and is not part of the definition of
\(\mathcal U_{\rm vac}\).

The nonlinear calculation gave the exact ultralocal bracket
\begin{equation}
        \{\mathcal C_0(x),\mathcal C_1(y)\}
        =
        \mathfrak d(x)\delta^{(2)}(x-y),
        \label{eq:scope-Delta-local}
\end{equation}
and at the maximally symmetric vacuum
\begin{equation}
        \mathfrak d_{\rm vac}
        =
        -\frac{\sqrt a}{2c}
        \neq0 .
        \label{eq:scope-d-vac}
\end{equation}
Continuity therefore guarantees an open neighborhood of the vacuum on
which the scalar pair remains second class.  The component
\(\mathcal U_{\rm vac}\) is the connected component of this regular
nonzero-rank region that contains the maximally symmetric vacuum.

The theorem does not assert that the rank remains unchanged everywhere in
phase space.  In particular, configurations satisfying
\begin{equation}
        \mathfrak d=0,
        \qquad
        s=0,
        \qquad\text{or}\qquad
        \det q_{\mu\nu}=0
        \label{eq:rank-changing-loci}
\end{equation}
must be analyzed separately.  At such loci the constraint structure may
change rank.

Nor have we proved that a Hamiltonian trajectory beginning in
\(\mathcal U_{\rm vac}\) can never reach its boundary.  Whether regular
initial data can evolve toward a rank-changing locus is a distinct global
dynamical question.  The result established here is that on every open
portion of a solution that remains inside \(\mathcal U_{\rm vac}\), the
local count is exactly two.

Because the scalar-rank calculation is local, any locally maximally
symmetric solution has the same nonzero vacuum value of
\(\mathfrak d\) at its regular points.  This includes locally
AdS\(_3\) quotients such as BTZ black holes \cite{BTZ}.  This local
statement does not, however, imply that all globally inequivalent
constant-curvature solutions belong to the same connected component of
the full phase space.

The linearized spectrum of Ref.~\cite{TekinAuxMetric2026} was not used in
the proof.  Agreement with its two massive spin-two polarizations is an
independent check.  Likewise, the present constraint count is not a
statement about stability, positivity of energy, tachyonic instabilities,
or boundary unitarity.  Those are separate questions that depend on the
chosen vacuum and boundary conditions.

The constraint structure is closely analogous to that found in
Hamiltonian analyses of New Massive Gravity
\cite{BlagojevicCvetkovic,SadeghShirzad,HohmRouthTownsendZhang,CSlike}:
a scalar constraint produces an independent secondary constraint, and
their nonvanishing mutual bracket removes the additional scalar canonical
pair.  We have not attempted to construct an explicit canonical
transformation between the first-order NMG variables used in those works
and the auxiliary-metric variables employed here.

What is special to the Born--Infeld formulation is the particularly useful
ADM structure of the auxiliary variables.  The mixed component \(v^i\)
enters algebraically and quadratically, while the normal component \(\nu\)
enters linearly.  This gives the starting point
\begin{equation}
        \mathcal C_0\weak0 .
\end{equation}
The crucial nonlinear information is then contained in the two identities
\begin{equation}
        \{\mathcal C_0,\mathcal C_0\}=0,
        \qquad
        \{\mathcal C_0(x),H_0[N]\}
        =
        N(x)\mathcal C_1(x),
\end{equation}
followed by the nonzero scalar rank
\begin{equation}
        \{\mathcal C_0(x),\mathcal C_1(y)\}
        =
        \mathfrak d(x)\delta^{(2)}(x-y),
        \qquad
        \mathfrak d\neq0 .
\end{equation}
The proof is therefore intrinsic to BINMG and does not rely on embedding
the theory into a parent bigravity model, in contrast with the motivation
of Ref.~\cite{PaulosTolley}.

Finally, the degree-of-freedom count itself is independent of the choice
\(\sigma=\pm1\).  Which sign yields a physically acceptable massive mode
about a particular vacuum is a spectral question studied at the
linearized level in Refs.~\cite{GSTBINMG,Gullu:2010st,TekinAuxMetric2026}.
The Dirac analysis determines how many local modes propagate, not whether
their energies have the desired sign.

\section{Conclusions}
\label{sec:conclusions}

We have carried out a fully nonlinear Hamiltonian analysis of
Born--Infeld New Massive Gravity in its auxiliary-metric formulation.  On
the vacuum-connected component \(\mathcal U_{\rm vac}\) of the regular
nonzero-rank region of phase space, BINMG propagates exactly two local
degrees of freedom per spatial point. The mechanism can be summarized succinctly.  The dynamical Legendre map
for \((\gamma_{ij},S^{ij})\) is nonsingular.  The mixed auxiliary component
\(v^i\) is algebraic and can be eliminated when \(\det S^{ij}\neq0\).  The normal auxiliary component \(\nu\) enters
linearly and imposes the scalar constraint \(\mathcal C_0\weak0\). The essential nonlinear step is the preservation of this constraint while
keeping the physical lapse arbitrary:
\begin{equation}
        \{\mathcal C_0(x),H_0[N]\}
        =
        N(x)\mathcal C_1(x).
\end{equation}
Thus consistency generates an independent secondary scalar constraint
rather than fixing the lapse.  Their mutual bracket is exactly
ultralocal,
\begin{equation}
        \{\mathcal C_0(x),\mathcal C_1(y)\}
        =
        \mathfrak d(x)\delta^{(2)}(x-y),
\end{equation}
and at the maximally symmetric BINMG vacuum
\begin{equation}
        \mathfrak d_{\rm vac}
        =
        -\frac{\sqrt a}{2c}
        \neq0 .
\end{equation}
Consequently \((\mathcal C_0,\mathcal C_1)\) is a second-class pair on
\(\mathcal U_{\rm vac}\), and preservation of
\(\mathcal C_1\) fixes the auxiliary variable \(\nu\), not the physical
lapse.

In the fully extended phase space there are twenty-four phase-space
dimensions, six first-class diffeomorphism constraints, and eight
second-class auxiliary constraints.  Hence
\begin{equation}
        24
        -
        2\times6
        -
        8
        =
        4
\end{equation}
physical phase-space dimensions, or
\begin{equation}
        N_{\rm dof}=2 .
\end{equation}

This nonlinear result agrees with the two massive spin-two polarizations
of the linearized theory
\cite{GSTBINMG,TekinAuxMetric2026}, with the known Hamiltonian structure of
NMG
\cite{BlagojevicCvetkovic,SadeghShirzad,HohmRouthTownsendZhang,CSlike},
and with the motivation from the bigravity construction of
Ref.~\cite{PaulosTolley}.  None of these results was used as an input to
the nonlinear constraint analysis.

The significance of this result is that BINMG is not a finite-order
higher-curvature theory.  The Born--Infeld determinant defines a
nonpolynomial theory which, when expanded in powers of the curvature,
contains terms of arbitrarily high order.  An analysis performed at any
finite order in that expansion would therefore not determine the
constraint structure of the full theory and, in particular, could not
exclude the appearance of additional nonlinear degrees of freedom at
higher orders.  The Hamiltonian analysis given here is instead performed
directly on the complete Born--Infeld action, without truncating the
curvature expansion.  It therefore determines the local nonlinear
degree-of-freedom content of the full theory to all orders: on
\(\mathcal U_{\rm vac}\), the infinite curvature series resums into a
theory with exactly two propagating local degrees of freedom and no
additional scalar mode.  This is highly nontrivial, since generic
higher-curvature corrections can alter the canonical structure and
introduce additional propagating modes.  The persistence of only two
local degrees of freedom is therefore a property of the complete
Born--Infeld structure, not something that could have been inferred from
any finite truncation.  In this sense, the result provides an exact
nonperturbative determination of the local mode content of an all-orders
higher-curvature gravity theory.

The theorem is local in phase space.  We have not classified the
rank-changing loci
\begin{equation}
        \det S^{ij}=0,
        \qquad
        \det q_{\mu\nu}=0,
        \qquad
        \mathfrak d=0,
\end{equation}
nor have we shown that trajectories beginning in
\(\mathcal U_{\rm vac}\) remain there globally.  The constraint count also
does not address stability, positivity of energy, tachyons, or boundary
unitarity. What has been established is more precise: on the regular
vacuum-connected component \(\mathcal U_{\rm vac}\), the full nonlinear
BINMG theory has exactly the two local degrees of freedom of a massive
spin-two field in three dimensions, with no additional scalar degree of
freedom.

\begin{acknowledgments}
Z.S.K. acknowledges financial support from the TUBITAK-BIDEB 2211
program.  B.T. would like to thank his former students
T.~\c{C}.~\c{S}i\c{s}man and \.I.~G\"ull\"u, with whom he introduced
BINMG and investigated many aspects of the theory, and M.~Do\u{g}ru,
whose master's thesis \cite{DogruThesis} pursued the constraint analysis
of the same theory in the purely metric formulation.  Although that
formulation did not lead to the complete nonlinear result obtained here,
the thesis was an important earlier exploration of the problem.
\end{acknowledgments}
\appendix

\section{Projection with respect to the auxiliary second-class constraints}
\label{app:improved-projection}

Before the auxiliary second-class constraints are eliminated, a convenient
representative of a diffeomorphism generator need not have weakly vanishing
Poisson brackets with them.  This does not signal a failure of the gauge
symmetry.  One may always replace the generator by an equivalent
representative whose Hamiltonian flow is tangent to the second-class
constraint surface.

Let
\begin{equation}
        \chi_A\weak0
        \label{eq:appendix-second-class-constraints}
\end{equation}
denote a set of second-class constraints.  The index \(A\) may include both
discrete labels and spatial position.  To distinguish their full Poisson
matrix from the scalar kernel
\(\Delta=\{\mathcal C_0,\mathcal C_1\}\) used in the main text, define
\begin{equation}
        \Omega_{AB}(x,y)
        :=
        \{\chi_A(x),\chi_B(y)\}.
        \label{eq:appendix-second-class-matrix}
\end{equation}
Second-classness means that, in a neighborhood of the constraint surface,
this kernel possesses an inverse \(\Omega^{AB}(x,y)\) satisfying
\begin{equation}
        \int\dd^2z\,
        \Omega^{AB}(x,z)\Omega_{BC}(z,y)
        =
        \delta^A{}_C\delta^{(2)}(x-y),
        \label{eq:appendix-Delta-inverse}
\end{equation}
and similarly with the order reversed.

We now use condensed notation, with integrations over repeated continuous
labels understood.  Given any phase-space functional \(F\), define its
projected representative by
\begin{equation}
        F^\star
        :=
        F
        -
        \{F,\chi_A\}
        \Omega^{AB}\chi_B .
        \label{eq:appendix-improved-F}
\end{equation}
Since the correction is proportional to the second-class constraints,
\begin{equation}
        F^\star\weak F .
        \label{eq:Fstar-weak-F}
\end{equation}
Thus \(F\) and \(F^\star\) define the same function on the
second-class surface.

More importantly, the flow generated by \(F^\star\) preserves that
surface.  Indeed,
\begin{align}
        \{F^\star,\chi_C\}
        =
        \{F,\chi_C\}
        -
        \{F,\chi_A\}
        \Omega^{AB}
        \{\chi_B,\chi_C\}-
        \bigl\{
        \{F,\chi_A\}\Omega^{AB},
        \chi_C
        \bigr\}\chi_B .
        \label{eq:appendix-improved-bracket-full}
\end{align}
The last term is explicitly proportional to \(\chi_B\).  Hence, on the
second-class surface,
\begin{align}
        \{F^\star,\chi_C\}
        &\weak
        \{F,\chi_C\}
        -
        \{F,\chi_A\}
        \Omega^{AB}\Omega_{BC}
        =
        \{F,\chi_C\}
        -
        \{F,\chi_C\}
        =
        0 .
        \label{eq:appendix-improved-proof}
\end{align}
Therefore
\begin{equation}
        \{F^\star,\chi_A\}\weak0 .
        \label{eq:Fstar-preserves-surface}
\end{equation}

This construction should not be confused with the origin of the gauge
symmetry itself.  The projection guarantees only that the chosen
representative has a flow tangent to the second-class surface.  The fact
that the normal transformation belongs to a diffeomorphism gauge chain
follows from spacetime covariance, independently of this construction.

Applied to the normal generator, this is the reason for writing
\begin{equation}
        \mathcal H_\perp^\star
        \weak
        \mathcal H_0 .
\end{equation}
The two functions agree on the auxiliary second-class surface, but
\(\mathcal H_\perp^\star\) is chosen so that its Poisson brackets with the
auxiliary second-class constraints vanish weakly.  The same construction
applies to the primary normal representative in the fully extended phase
space.

The equivalent description uses the Dirac bracket.  For the same
second-class set,
\begin{equation}
        \{F,G\}_{\rm D}
        =
        \{F,G\}
        -
        \{F,\chi_A\}
        \Omega^{AB}
        \{\chi_B,G\}.
        \label{eq:appendix-Dirac-bracket}
\end{equation}
By construction,
\begin{equation}
        \{\chi_A,F\}_{\rm D}=0
        \label{eq:Dirac-chi-strong-zero}
\end{equation}
strongly for every \(F\).  The second-class constraints may therefore be
imposed strongly after passing to the Dirac bracket:
\begin{equation}
        \chi_A=0 .
\end{equation}

Thus there are two equivalent descriptions.  One may retain the ordinary
Poisson bracket and use projected representatives such as \(F^\star\), or
one may eliminate the second-class variables and work with the Dirac
bracket.  Both descriptions induce the same dynamics on the physical
phase space and give the same degree-of-freedom count.

\section{Independent check in the diagonal ultralocal sector}
\label{app:exact-ultralocal}

We give here an independent check of the canonical formulas used in the
main text.  The calculation is deliberately restricted to a spatially
homogeneous diagonal sector and performs the Legendre transformation
directly from the spacetime metric components, without using the ADM
Legendre formulas derived earlier.

This appendix is only a check.  Since the lapse is fixed to unity and all
spatial gradients are suppressed, it cannot establish either the
arbitrary-lapse identity
\begin{equation}
        \{\mathcal C_0(x),H_0[N]\}
        =
        N(x)\mathcal C_1(x)
\end{equation}
of Proposition~\ref{prop:arbitrary-lapse-factorization}, or the exact
absence of derivatives of the delta function in
\(\{\mathcal C_0,\mathcal C_1\}\), proved in
Proposition~\ref{prop:Delta-exact-ultralocal}.

We choose
\begin{equation}
        N=1,
        \qquad
        N^i=0,
\end{equation}
and assume that all spatial fields depend only on time.  Let
\begin{equation}
        \gamma_{ij}
        =
        \operatorname{diag}(g_1,g_2),
        \qquad
        S^{ij}
        =
        \operatorname{diag}(s_1,s_2).
        \label{eq:appendix-diagonal-fields}
\end{equation}
Spatial homogeneity gives
\begin{equation}
        B_i=D_jp^j{}_i=0,
\end{equation}
so the general auxiliary solution
\begin{equation}
        v_*^i
        =
        -\frac2sS^{ij}B_j
\end{equation}
reduces consistently to
\begin{equation}
        v^i=0 .
\end{equation}

The spacetime metric and auxiliary density are therefore
\begin{equation}
        g_{\mu\nu}
        =
        \operatorname{diag}(-1,g_1,g_2),
        \qquad
        P^{\mu\nu}
        =
        \operatorname{diag}(\nu,s_1,s_2).
        \label{eq:appendix-diagonal-spacetime}
\end{equation}

A direct calculation gives
\begin{equation}
        G_{00}
        =
        \frac{\dot g_1\dot g_2}{4g_1g_2},
        \label{eq:G00-diagonal}
\end{equation}
and
\begin{equation}
        G_{11}
        =
        \frac{g_1}{4g_2^2}
        \left(
        \dot g_2^{\,2}
        -
        2g_2\ddot g_2
        \right),
        \label{eq:G11-diagonal}
\end{equation}
\begin{equation}
        G_{22}
        =
        \frac{g_2}{4g_1^2}
        \left(
        \dot g_1^{\,2}
        -
        2g_1\ddot g_1
        \right).
        \label{eq:G22-diagonal}
\end{equation}
Thus \(G_{11}\) contains \(\ddot g_2\), but not \(\ddot g_1\), whereas
\(G_{22}\) contains \(\ddot g_1\), but not \(\ddot g_2\).  This is the
direct component manifestation of the trace subtraction in the spatial
Einstein projection.  Since
\begin{equation}
        K_{ij}
        =
        \frac12
        \operatorname{diag}(\dot g_1,\dot g_2),
\end{equation}
the ADM expression for \(G_{ij}\) used in the main text reproduces
Eqs.~\eqref{eq:G11-diagonal} and \eqref{eq:G22-diagonal}, including all
signs.

Substitution into the polynomial auxiliary density
\eqref{eq:aux-P-action-expanded-pedagogical} gives
\begin{align}
        \mathcal L
        ={}&
        \frac12
        \left(
        -\nu+g_1s_1+g_2s_2
        \right)
        +
        c
        \left(
        \nu G_{00}
        +
        s_1G_{11}
        +
        s_2G_{22}
        \right)
        \nonumber\\
        &+
        \frac12\nu s_1s_2
        -
        \beta\sqrt{g_1g_2}.
        \label{eq:L-diagonal-second-order}
\end{align}
Integrating the terms containing \(\ddot g_a\) by parts and discarding the
resulting total time derivative gives
\begin{align}
        \mathcal L_{\rm ul}
        ={}&
        \frac12(-\nu+g_1s_1+g_2s_2)
        +
        \frac12\nu s_1s_2
        -
        \beta\sqrt{g_1g_2}
        +
        \frac{c\nu}{4g_1g_2}
        \dot g_1\dot g_2
        \nonumber\\
        &+
        \frac c2
        \frac{\dd}{\dd t}
        \left(
        \frac{s_1g_1}{g_2}
        \right)
        \dot g_2
        +
        \frac{cs_1g_1}{4g_2^2}
        \dot g_2^{\,2}
        \nonumber\\
        &+
        \frac c2
        \frac{\dd}{\dd t}
        \left(
        \frac{s_2g_2}{g_1}
        \right)
        \dot g_1
        +
        \frac{cs_2g_2}{4g_1^2}
        \dot g_1^{\,2}.
        \label{eq:L-ultralocal-first-order}
\end{align}

Let \(r_a\) denote the momentum conjugate to \(g_a\), and \(p_a\) the
momentum conjugate to \(s_a\).  The latter satisfy
\begin{equation}
        p_1
        =
        \frac{cg_1}{2g_2}\dot g_2,
        \qquad
        p_2
        =
        \frac{cg_2}{2g_1}\dot g_1.
        \label{eq:p-diagonal-velocities}
\end{equation}
Hence
\begin{equation}
        \dot g_1
        =
        \frac{2g_1p_2}{cg_2},
        \qquad
        \dot g_2
        =
        \frac{2g_2p_1}{cg_1}.
        \label{eq:gdot-diagonal-solved}
\end{equation}
This directly confirms the invertibility of the mixed velocity block for
\(c\neq0\).

The remaining momentum equations,
\begin{equation}
        r_a
        =
        \frac{\partial\mathcal L_{\rm ul}}
             {\partial\dot g_a},
\end{equation}
are linear in \(\dot s_1,\dot s_2\).  Solving them gives
\begin{align}
        \dot s_1
        =
        \frac{1}{cg_1g_2}
        \bigl(
        &-2g_1p_2s_1
        +
        2g_2^2r_2
        +
        2g_2p_1s_1
        -
        2g_2p_2s_2
        -
        \nu p_2
        \bigr),
        \label{eq:s1dot-solved}
\end{align}
and
\begin{align}
        \dot s_2
        =
        \frac{1}{cg_1g_2}
        \bigl(
        &2g_1^2r_1
        -
        2g_1p_1s_1
        +
        2g_1p_2s_2
        -
        2g_2p_1s_2
        -
        \nu p_1
        \bigr).
        \label{eq:s2dot-solved}
\end{align}

The Legendre transform is
\begin{equation}
        \mathcal H_{\rm ul}
        =
        r_1\dot g_1+r_2\dot g_2
        +
        p_1\dot s_1+p_2\dot s_2
        -
        \mathcal L_{\rm ul}.
\end{equation}
After substituting the velocity solutions, the dependence on \(\nu\)
remains linear:
\begin{equation}
        \mathcal H_{\rm ul}
        =
        \mathcal H_0^{\rm ul}
        +
        \nu\mathcal C_0^{\rm ul}.
        \label{eq:H-ultralocal-separation}
\end{equation}
The reduced unit-lapse Hamiltonian density is
\begin{align}
        \mathcal H_0^{\rm ul}
        ={}&
        \beta\sqrt{g_1g_2}
        -
        \frac12(g_1s_1+g_2s_2)
        \nonumber\\
        &+
        \frac{2}{cg_1g_2}
        \left[
        g_1^2p_2r_1
        +
        g_2^2p_1r_2
        -
        g_1p_1p_2s_1
        -
        g_2p_1p_2s_2
        \right.
        \nonumber\\
        &\hspace{3.7cm}\left.
        +
        \frac12g_1p_2^2s_2
        +
        \frac12g_2p_1^2s_1
        \right],
        \label{eq:H0-ultralocal-exact}
\end{align}
while
\begin{equation}
        \mathcal C_0^{\rm ul}
        =
        \frac12(1-s_1s_2)
        -
        \frac{p_1p_2}{cg_1g_2}.
        \label{eq:C0-ultralocal-exact}
\end{equation}

In this homogeneous sector,
\begin{equation}
        B_i=0,
        \qquad
        D_iD_j\mathcal A^{ij}=0,
        \qquad
        {}^{(2)}R=0.
\end{equation}
Direct substitution of the diagonal canonical variables into
Eqs.~\eqref{eq:C0-before-v-elimination} and  \eqref{eq:H0-full-tensor} reproduces
Eqs.~\eqref{eq:H0-ultralocal-exact} and
\eqref{eq:C0-ultralocal-exact}.  This is a nontrivial check because the
calculation above began directly from the spacetime metric rather than
from the ADM Legendre transformation.

The reduced diagonal Poisson bracket is
\begin{equation}
        \{F,G\}_{\rm diag}
        =
        \sum_{a=1}^2
        \left(
        \frac{\partial F}{\partial g_a}
        \frac{\partial G}{\partial r_a}
        -
        \frac{\partial F}{\partial r_a}
        \frac{\partial G}{\partial g_a}
        +
        \frac{\partial F}{\partial s_a}
        \frac{\partial G}{\partial p_a}
        -
        \frac{\partial F}{\partial p_a}
        \frac{\partial G}{\partial s_a}
        \right).
        \label{eq:diagonal-Poisson-bracket}
\end{equation}
At unit lapse define
\begin{equation}
        \mathcal C_1^{\rm ul}
        :=
        \{\mathcal C_0^{\rm ul},
        \mathcal H_0^{\rm ul}\}_{\rm diag}.
\end{equation}
A direct calculation gives
\begin{align}
        \mathcal C_1^{\rm ul}
        ={}&
        \frac{1}{2c^2g_1^2g_2^2}
        \Bigl[
        -2cg_1^3g_2r_1s_1
        +
        2cg_1^2g_2p_1s_1^2
        -
        cg_1^2g_2p_2
        \nonumber\\
        &\hspace{2.2cm}
        -
        2cg_1g_2^3r_2s_2
        -
        cg_1g_2^2p_1
        +
        2cg_1g_2^2p_2s_2^2
        \nonumber\\
        &\hspace{2.2cm}
        +
        2g_1p_1p_2^2
        +
        2g_2p_1^2p_2
        \Bigr].
        \label{eq:C1-ultralocal-exact}
\end{align}
Taking one further bracket,
\begin{equation}
        \Delta_{\rm ul}
        :=
        \{\mathcal C_0^{\rm ul},
        \mathcal C_1^{\rm ul}\}_{\rm diag},
\end{equation}
gives
\begin{align}
        \Delta_{\rm ul}
        ={}&
        -\frac{1}{4c^2g_1^2g_2^2}
        \Bigl[
        2cg_1^2g_2s_1^2s_2
        -
        cg_1^2g_2s_1
        \nonumber\\
        &\hspace{2.2cm}
        +
        2cg_1g_2^2s_1s_2^2
        -
        cg_1g_2^2s_2
        \nonumber\\
        &\hspace{2.2cm}
        +
        4g_1^2p_2r_1
        +
        2g_1p_2^2s_2
        +
        4g_2^2p_1r_2
        +
        2g_2p_1^2s_1
        \Bigr].
        \label{eq:Delta-ultralocal-exact}
\end{align}
Thus the scalar cross-bracket is not identically zero even in this
restricted sector.

There is one qualification.  At a diagonal configuration, reflection of
either spatial basis vector changes the sign of off-diagonal tensor
components while leaving scalar expressions invariant.  Therefore the
first derivative of any scalar invariant with respect to an off-diagonal
component vanishes at the diagonal configuration.  The omitted
off-diagonal canonical pairs consequently do not contribute to the purely
algebraic scalar bracket evaluated here.

This argument does not apply to spatial-derivative terms.  Their absence
from the diagonal homogeneous model is precisely why the exact
ultralocality result of
Proposition~\ref{prop:Delta-exact-ultralocal} was required in the full
theory.

As a simple witness point, take
\begin{equation}
        g_1=g_2=1,
        \qquad
        s_1=s_2=1,
        \qquad
        p_1=p_2=r_1=r_2=0 .
        \label{eq:diagonal-witness-point}
\end{equation}
Then
\begin{equation}
        \mathcal C_0^{\rm ul}=0,
        \qquad
        \mathcal C_1^{\rm ul}=0,
        \qquad
        \Delta_{\rm ul}
        =
        -\frac1{2c}.
        \label{eq:diagonal-witness-bracket}
\end{equation}
For \(\beta=1\),
\begin{equation}
        a=\beta^2=1,
        \qquad
        2c\Lambda=1-a
\end{equation}
implies \(\Lambda=0\), so
Eq.~\eqref{eq:diagonal-witness-point} is the spatial part of the flat
BINMG vacuum.  The result agrees exactly with the general vacuum formula
\begin{equation}
        \mathfrak d_{\rm vac}
        =
        -\frac{\sqrt a}{2c}
        =
        -\frac1{2c}.
\end{equation}
For \(\beta\neq1\), the time-symmetric maximally symmetric vacuum slice is
curved and therefore does not belong to the spatially homogeneous flat
sector considered here.  Its scalar rank is instead given by
Sec.~\ref{sec:vacuum-scalar-rank}.

Thus the direct diagonal calculation independently checks the Legendre
map, the reduced Hamiltonian, the scalar constraint, and the nonvanishing
of the scalar Dirac bracket in the flat-vacuum limit.  The arbitrary-lapse
factorization and exact ultralocality of the full theory remain the
stronger results proved in the main text.

\section{Degenerate and rank-changing cases}
\label{app:excluded-branches}

The degree-of-freedom theorem applies only where the Hamiltonian system has
the regularity and constant-rank properties assumed in its derivation.
It is useful to distinguish genuine rank-changing loci from singular
parameter limits and singular choices of ADM variables.

\begin{enumerate}[label=\textup{(\roman*)},leftmargin=2em]

\item \textit{\(c=0\).}

The mixed block of the dynamical velocity Hessian is proportional to
\begin{equation}
        c=\frac{\sigma}{2m^2}.
\end{equation}
For BINMG with finite \(m^2\) and \(\sigma=\pm1\), one has \(c\neq0\).
The formal limit \(c\to0\), corresponding to \(m^2\to\infty\), is singular
in the canonical variables used in this paper: the block responsible for
the invertibility of the Legendre map disappears.  This limiting theory
must therefore be analyzed separately in variables adapted to that limit.

\item \textit{Singular ADM charts.}

The ADM decomposition assumes a nondegenerate spatial metric,
\begin{equation}
        \det\gamma_{ij}\neq0,
\end{equation}
and a coordinate patch in which
\begin{equation}
        N\neq0 .
\end{equation}
If either condition fails, the normal--tangential decomposition itself is
singular.  In particular, \(N=0\) is not a physical branch condition and
does not signal a change in the number of degrees of freedom.  It simply
lies outside the regular ADM chart used in the analysis and should not be
confused with a dynamical rank-changing condition such as
\(\mathfrak d=0\).

\item \textit{Degenerate auxiliary metric.}

The auxiliary-metric formulation requires
\begin{equation}
        \det q_{\mu\nu}\neq0 .
\end{equation}
If \(q_{\mu\nu}\) degenerates, the inverse \(q^{\mu\nu}\) and hence the
definition
\begin{equation}
        P^{\mu\nu}
        =
        \sqrt{-q}\,q^{\mu\nu}
\end{equation}
are no longer available.  The equivalence between the auxiliary-metric and
determinant formulations used in this paper was established only on the
nondegenerate sector.  We make no claim about configurations with
\(\det q_{\mu\nu}=0\).

\item \textit{Degenerate spatial auxiliary block.}

The equation determining the mixed auxiliary component is most usefully
written in a form that remains meaningful even when \(S^{ij}\) becomes
singular.  Define
\begin{equation}
        M_{ij}
        :=
        \operatorname{adj}(S)_{ij}.
        \label{eq:degenerate-v-Hessian}
\end{equation}
On the regular region \(s=\det S^{ij}\neq0\),
\begin{equation}
        M_{ij}
        =
        s(S^{-1})_{ij},
\end{equation}
and the \(v^i\) equation is
\begin{equation}
        2B_i+M_{ij}v^j=0 .
\end{equation}
In two spatial dimensions,
\begin{equation}
        \det M=s .
        \label{eq:det-v-Hessian}
\end{equation}
Therefore the Hessian in the \(v^i\) sector is invertible precisely when
\begin{equation}
        s=\det S^{ij}\neq0 .
\end{equation}
On this regular region,
\begin{equation}
        v_*^i
        =
        -\frac2sS^{ij}B_j .
\end{equation}

The advantage of writing \(M_{ij}\) as the adjugate is that
Eq.~\eqref{eq:degenerate-v-Hessian} itself remains well defined at
\(s=0\).  There \(M_{ij}\) is singular, however, and \(v^i\) can no longer
be solved for uniquely.  The vector second-class block therefore changes
rank, and the unreduced Dirac system must be analyzed directly.

\item \textit{Vanishing scalar Dirac coefficient.}

On the regular branch,
\begin{equation}
        \{\mathcal C_0(x),\mathcal C_1(y)\}
        =
        \mathfrak d(x)\delta^{(2)}(x-y).
\end{equation}
The theorem assumes
\begin{equation}
        \mathfrak d(x)\neq0 .
\end{equation}
At a point where \(\mathfrak d=0\), the scalar Dirac kernel is no longer
invertible.  Preservation of \(\mathcal C_1\) then need not determine
\(\nu\).  Depending on the remaining brackets, the Dirac algorithm could
produce an additional constraint, acquire an additional gauge degeneracy,
or enter a distinct constraint branch.  The regular analysis does not
decide which possibility is realized.

\end{enumerate}

These cases delimit the theorem proved in the main text.  The statement
\begin{equation}
        N_{\rm dof}=2
\end{equation}
has been established on \(\mathcal U_{\rm vac}\), where the auxiliary
formulation and the ADM decomposition are regular, the dynamical Legendre
map is invertible, and the relevant Dirac matrices have constant rank.
No claim is made about the singular or rank-changing loci listed above.


\begin{thebibliography}{99}



\bibitem{BTZ}
M.~Ba\~nados, C.~Teitelboim, and J.~Zanelli,
``The black hole in three-dimensional space-time,''
Phys. Rev. Lett. \textbf{69}, 1849 (1992).

\bibitem{Witten}
E.~Witten,
``(2+1)-dimensional gravity as an exactly soluble system,''
Nucl. Phys. B \textbf{311}, 46 (1988).

\bibitem{DeserJackiwTempleton}
S.~Deser, R.~Jackiw, and S.~Templeton,
``Topologically massive gauge theories,''
Ann. Phys. (N.Y.) \textbf{140}, 372 (1982).

\bibitem{BHT2009}
E.~A.~Bergshoeff, O.~Hohm, and P.~K.~Townsend,
``Massive gravity in three dimensions,''
Phys. Rev. Lett. \textbf{102}, 201301 (2009).

\bibitem{BHTMore}
E.~A.~Bergshoeff, O.~Hohm, and P.~K.~Townsend,
``More on massive 3D gravity,''
Phys. Rev. D \textbf{79}, 124042 (2009).

\bibitem{GulluTek}
\.I.~G\"ull\"u and B.~Tekin,
``Massive higher derivative gravity in $D$-dimensional anti-de Sitter
spacetimes,''
Phys. Rev. D \textbf{80}, 064033 (2009).

\bibitem{GSTCanonical}
\.I.~G\"ull\"u, T.~\c{C}.~\c{S}i\c{s}man, and B.~Tekin,
``Canonical structure of higher derivative gravity in 3D,''
Phys. Rev. D \textbf{81}, 104017 (2010).

\bibitem{Tekin2016}
B.~Tekin,
``Particle content of quadratic and
$f(R_{\mu\nu\sigma\rho})$ theories in $(A)dS$,''
Phys. Rev. D \textbf{93}, 101502(R) (2016).

\bibitem{BoulwareDeser}
D.~G.~Boulware and S.~Deser,
``Can gravitation have a finite range?''
Phys. Rev. D \textbf{6}, 3368 (1972).

\bibitem{MandM}
S.~Deser and B.~Tekin,
``Massive, topologically massive, models,''
Class. Quantum Grav. \textbf{19}, L97 (2002).

\bibitem{GSTBINMG}
\.I.~G\"ull\"u, T.~\c{C}.~\c{S}i\c{s}man, and B.~Tekin,
``Born--Infeld extension of new massive gravity,''
Class. Quantum Grav. \textbf{27}, 162001 (2010).

\bibitem{Sinha}
A.~Sinha,
``On the new massive gravity and AdS/CFT,''
JHEP \textbf{06}, 061 (2010).

\bibitem{Gullu:2010st}
\.I.~G\"ull\"u, T.~\c{C}.~\c{S}i\c{s}man, and B.~Tekin,
``$c$-functions in the Born--Infeld extended new massive gravity,''
Phys. Rev. D \textbf{82}, 024032 (2010).

\bibitem{Paulos0}
M.~F.~Paulos,
``New massive gravity extended with an arbitrary number of curvature
corrections,''
Phys. Rev. D \textbf{82}, 084042 (2010).

\bibitem{TahsinTez}
T.~\c{C}.~\c{S}i\c{s}man,
\textit{Born-Infeld Gravity Theories in $D$ Dimensions},
Ph.D. thesis, Middle East Technical University, Ankara (2012).

\bibitem{Dirac}
P.~A.~M.~Dirac,
\textit{Lectures on Quantum Mechanics}
(Yeshiva University, New York, 1964).

\bibitem{Castellani}
L.~Castellani,
``Symmetries in constrained Hamiltonian systems,''
Ann. Phys. (N.Y.) \textbf{143}, 357 (1982).

\bibitem{HenneauxTeitelboim}
M.~Henneaux and C.~Teitelboim,
\textit{Quantization of Gauge Systems}
(Princeton University Press, Princeton, 1992).

\bibitem{BlagojevicCvetkovic}
M.~Blagojevi\'c and B.~Cvetkovi\'c,
``Hamiltonian analysis of BHT massive gravity,''
JHEP \textbf{01}, 082 (2011).

\bibitem{SadeghShirzad}
M.~Sadegh and A.~Shirzad,
``Constraint structure of the three-dimensional massive gravity,''
Phys. Rev. D \textbf{83}, 084040 (2011).

\bibitem{HohmRouthTownsendZhang}
O.~Hohm, A.~Routh, P.~K.~Townsend, and B.~Zhang,
``On the Hamiltonian form of 3D massive gravity,''
Phys. Rev. D \textbf{86}, 084035 (2012).

\bibitem{CSlike}
E.~A.~Bergshoeff, O.~Hohm, W.~Merbis, A.~J.~Routh, and P.~K.~Townsend,
``Chern--Simons-like gravity theories,''
Lect. Notes Phys. \textbf{892}, 181 (2015).

\bibitem{AfsharBergshoeffMerbis}
H.~R.~Afshar, E.~A.~Bergshoeff, and W.~Merbis,
``Extended massive gravity in three dimensions,''
JHEP \textbf{08}, 115 (2014).

\bibitem{PaulosTolley}
M.~F.~Paulos and A.~J.~Tolley,
``Massive gravity theories and limits of ghost-free bigravity models,''
JHEP \textbf{09}, 002 (2012).

\bibitem{deRhamNMG}
C.~de Rham, G.~Gabadadze, D.~Pirtskhalava, A.~J.~Tolley, and I.~Yavin,
``Nonlinear dynamics of 3D massive gravity,''
JHEP \textbf{06}, 028 (2011).

\bibitem{TekinAuxMetric2026}
B.~Tekin,
``Auxiliary-metric formulation of Born-Infeld new massive gravity,''
Phys. Rev. D \textbf{114}, 044025 (2026).

\bibitem{ADM}
R.~Arnowitt, S.~Deser, and C.~W.~Misner,
``The dynamics of general relativity,''
in \textit{Gravitation: An Introduction to Current Research},
edited by L.~Witten
(Wiley, New York, 1962), Chap.~7;
reprinted in Gen. Relativ. Gravit. \textbf{40}, 1997 (2008).

\bibitem{Gourgoulhon}
E.~Gourgoulhon,
``3+1 formalism and bases of numerical relativity,''
arXiv:gr-qc/0703035.

\bibitem{DogruThesis}
M.~Do\u{g}ru,
\textit{ADM Formulation of Generic Massless Spin-2 Gravity},
M.Sc. thesis, Middle East Technical University, Ankara (2018).

\bibitem{non}
E.~Altas and B.~Tekin,
``Nonstationary energy in general relativity,''
Phys. Rev. D \textbf{101}, 024035 (2020).

\bibitem{DeserTekin}
S.~Deser and B.~Tekin,
``Energy in generic higher curvature gravity theories,''
Phys. Rev. D \textbf{67}, 084009 (2003).

\end{thebibliography}
\end{document}